\documentclass[lettersize,journal]{IEEEtran}

\usepackage{packages}
\usepackage{settings}
\usepackage{commands_zh}

\usepackage{eso-pic}
\newlength{\ieeenoticepullup}
\newcommand{\ieeenoticeadjcol}{\enlargethispage{-\ieeenoticepullup}}
\newcommand{\ieeenotice}{%
  \AddToShipoutPictureFG*{%
    \AtPageLowerLeft{%
      \put(\LenToUnit{\dimexpr1in+\oddsidemargin\relax},\LenToUnit{0.4in}){%
        \makebox[0pt][l]{\parbox[b]{\textwidth}{\fontsize{6.6}{7.6}\selectfont
          \copyright~2026 IEEE. Personal use of this material is permitted. Permission from IEEE must be obtained for all other uses, in any current or future media, including reprinting/republishing this material for advertising or promotional purposes, creating new collective works, for resale or redistribution to servers or lists, or reuse of any copyrighted component of this work in other works.\\
          Accepted version of: Z.~Hu, F.~Mammadova, J.~Tachella, M.~Unser, and J.~Dong, ``Structured Random Models for Phase Retrieval With Optical Diffusers,'' \emph{IEEE Trans.\ Comput.\ Imaging}, vol.~12, pp.~1298--1309, 2026. DOI: \href{https://doi.org/10.1109/TCI.2026.3711006}{10.1109/TCI.2026.3711006}.}}}}}}

\begin{document}
\ieeenotice

\title{Structured Random Models for Phase Retrieval with Optical Diffusers}
\author{Zhiyuan Hu, Fakhriyya Mammadova, Julian Tachella, Michael Unser, Jonathan Dong
	\thanks{This work was supported by the Swiss National Science Foundation under Grant
		PZ00P2\_216211. Julian Tachella acknowledges support by the French ANR grant UNLIP (ANR-23-CE23-0013). \textit{(Corresponding Author: Zhiyuan Hu.)} \\

		Zhiyuan Hu, Fakhriyya Mammadova, Michael Unser and Jonathan Dong are with the Biomedical Imaging Group, École Polytechnique Fédérale de Lausanne, 1015, Switzerland. (e-mail: zhiyuan.hu@epfl.ch; faxriyya489@kaist.ac.kr; michael.unser@epfl.ch; jonathan.dong@epfl.ch) \\

		%Fakhriyya Mammadova is with the School of Electrical Engineering, Korea Advanced Institute of Science and Technology, 34141, Republic of Korea. (email: faxriyya489@kaist.ac.kr) \\

		Julian Tachella is with LPENSL, CNRS \& ENS Lyon, 46 allée de l'Italie, 69364, France. (e-mail: julian.tachella@cnrs.fr)

	}}

\maketitle
\ieeenoticeadjcol % room for the IEEE notice (1st column of page 1)

\begin{abstract}
	Phase retrieval is a nonlinear inverse problem that arises in a wide range of imaging modalities, from electron microscopy to Fourier ptychography.
	In particular, the reconstruction is facilitated when the sensing matrix is i.i.d. random, enabling strong theoretical guarantees and efficient reconstruction algorithms.
	However, its applicability is restricted by excessive computational costs.
	In this paper, we propose structured random models for phase retrieval, where we emulate a dense random matrix by a cascade of structured transforms and random diagonal matrices.
	We reduce the complexity from quadratic to log-linear at no cost in reconstruction performance.
	Through a spectral method initialization followed by gradient descent, robust reconstruction is obtained at an oversampling ratio as low as 2.5.
	Moreover, we observe that the reconstruction performance is solely determined by the singular-value distribution of the forward matrix.
	This class of models can directly be implemented with basic optical elements such as lenses and diffusers, paving the way for large-scale phase imaging with robust reconstruction guarantees.
\end{abstract}


\section{Introduction} \label{sec:intro}
Phase retrieval refers to reconstructing the phase information of a complex-valued field from intensity-only measurements.
%, which are obtained through a linear transformation of the original object.
It enables the complete characterization of the electric field for label-free imaging and wavefront engineering and pertains to modalities such as astronomy~\cite{fienup1987phase}, crystallography~\cite{sayre1952some}, computer-generated holography~\cite{zhang20173d,eybposh2020deepcgh}, optical computing~\cite{gupta2019don}, Fourier ptychography~\cite{yeh2015experimental}, electron and optical microscopy~\cite{zheng2013wide,boniface2020non}.
However, the nonlinearity introduced by intensity-only measurements significantly complicates the reconstruction problem compared to linear models.
One key parameter to quantify the reconstruction difficulty is the oversampling ratio (OR), defined as the ratio between the dimension of the measurements and that of the unknown signal.
While a high OR facilitates reconstruction, it comes at the cost of additional measurements. Thus, reconstruction under a low OR is desirable.

Existing phase-imaging modalities typically involve Fourier transforms~\cite{dong2023phase}.
The earliest occurrence, Fourier phase retrieval~\cite{fienup1982phase}, reconstructs an image based solely on its Fourier magnitude.
This challenging reconstruction problem is made easier when more measurements become available, for instance, through shifts of the source position in ptychography~\cite{paxman1992joint,rodenburg2019ptychography} or through several illumination angles in Fourier ptychography~\cite{yeh2015experimental,bian2015fourier}.
These systems result in a sensing matrix composed of Fourier and diagonal matrices, which enables high-resolution imaging thanks to efficient computation.
Despite the observation that more measurements lead to an inverse problem that is easier to solve, theoretical guarantees are lacking to ensure proper reconstructions, even given many measurements.
Additionally, algorithms can become trapped in a suboptimal local minimum due to the nonlinear optimization landscape.

\ieeenoticeadjcol % room for the IEEE notice (2nd column of page 1)
Meanwhile, random phase retrieval~\cite{metzler2017coherent}, where the forward operator is an i.i.d. random matrix, benefits from a rich theoretical foundation.
Previous studies have established information-theoretic bounds on the reconstruction accuracy as a function of the OR~\cite{mondelli2019fundamental, maillard2020phase}, deriving a perfect-recovery threshold at an OR slightly above 2.
In practice, efficient dedicated algorithms have also been proposed, for instance, approximate message passing (AMP)~\cite{metzler2017coherent,maillard2020phase,schniter2016vector} and spectral methods~\cite{guizar2008phase,ma2021spectral,luo2019optimal}.
Although corresponding to the physical propagation of waves in a complex medium~\cite{metzler2017coherent,popoff2010measuring}, the random model is rarely implemented in practice due to its prohibitively high computational cost.
The $\mathcal{O}(n^2)$ dense matrix-vector multiplication required to evaluate the forward model makes reconstruction intractable in both space and time for high-resolution imaging.

Interestingly, many recent examples converge toward the replacement of a dense matrix by a fast linear transform.
A linear transform of the form $\mathbf{H}\mathbf{D}$ is employed in compressed sensing \cite{oymak2018isometric} and locality-sensitive hashing \cite{ailon2009fast}.
It corresponds to a subsampled Hadamard transform associated with a random diagonal matrix, an efficient structured transform that can be evaluated in $\mathcal{O}(n \log n)$.
In nonlinear models such as neural networks~\cite{moczulski2016acdc, dong2020reservoir,d2025comparison} or random features \cite{yu2016orthogonal}, a sequence of efficient transforms (e.g., $\mathbf{HDHD}$) is used to emulate dense connections.
These transforms occur naturally in optics, as they correspond to operations performed by lenses and planar interfaces.
Moreover, the multiplication by a random diagonal matrix can be implemented with optical diffusers.
They have been introduced in lensless imaging~\cite{antipa2017diffusercam, boominathan2020phlatcam, bezzam2025towards}; yet, their applicability to phase imaging has not been investigated.
% Besides,
% theoretical justifications for this substitution have recently emerged: \cite{wang2024universality} proved that AMP algorithms for compressed sensing exhibit identical state evolution when the sensing matrix is an i.i.d. Gaussian matrix or a random diagonal matrix wrapped by two Fourier transforms. Furthermore, \cite{dudeja2024spectral} established a spectral universality class for matrices of the form $\mathbf{JS}$ with a deterministic matrix $\mathbf{J}$ and a random sign-flipping diagonal matrix $\mathbf{S}$, to obtain the identical performance of i.i.d. sensing matrices in linear inverse problems. Building on this evidence from linear inverse problems, the validity of these designs for the nonlinear phase retrieval problem with gradient-based reconstruction algorithms remains to be established.
% With the evidences from linear inverse problems, the use of these designs is yet to be validated for the nonlinear phase retrieval problem with gradient-based reconstruction algorithms.
Besides, theoretical justifications for structured substitutions have recently emerged for \emph{linear} inverse problems. It is proven in \cite{wang2024universality} that the state evolution of approximate message passing is universal over broad classes of structured matrices, characterized by their asymptotic singular-value distribution. 
% including subsampled Fourier or Hadamard transforms composed with diagonal and random signed permutations,
% which behave like a rotationally-invariant matrix sharing the same limiting singular-value distribution.
Relatedly, \cite{dudeja2024spectral} introduced a spectral universality class for matrices $\mathbf{JS}$, with $\mathbf{J}$ deterministic and $\mathbf{S}$ a random sign-flipping diagonal, and showed that the performance of regularized least-squares estimators depends on the singular-value spectrum alone, independently of the specific construction.
These results strongly motivate replacing dense random matrices with structured operators; however, they concern linear measurements, and their validity for the nonlinear phase retrieval problem with gradient-based reconstruction remains to be established.

In this paper, we propose a structured random model to facilitate phase retrieval through a cascade of structured transforms and random diagonal matrices~\cite{hu2025structured}.
We demonstrate that the structured random model enables the same reconstruction quality as dense random models, and is much more computationally efficient.
We present that a model architecture composed of two structured transforms and diagonal operations is the minimal setting in which a dense random matrix can be emulated.
We also show that the singular-value distribution of the model solely determines the reconstruction performance.
In practice, this system can be implemented optically with lenses and diffusers.
This framework is highly flexible, as the results remain unchanged with other diffuser phase distributions and structured transforms.
% $0/\pi$ phase-shift diffusers and other structured transforms.

The organization of this paper is as follows: in Section~\ref{sec:srpr}, we define structured random models for phase retrieval and the corresponding optical systems, and analyze the covariance patterns of the elements in the sensing matrices. In Section~\ref{sec:methods}, we provide more details on the reconstruction configuration, including algorithms (gradient descent and spectral methods) and evaluation metrics. In Section~\ref{sec:results}, we present the evaluation results of the structured random model.

\section{Structured Random Model} \label{sec:srpr}
In this section, we show that at least 2 structured transforms and random diagonal matrices are needed to obtain a sensing matrix with uncorrelated elements.

\subsection{Model Definition}

The forward model of phase retrieval is
\begin{equation}\label{def:phase_retrieval}
	\measure = | \sensing \signal |^2,
\end{equation}
with $\signal \in \mathbb{C}^n$ the signal to be reconstructed, $\sensing \in \mathbb{C}^{m \times n}$ the sensing matrix, and $\measure \in \mathbb{R}^m$ the measurements.
We adopt the squared-modulus convention as it directly reflects the intensity measurements recorded by optical detectors \cite{yeh2015experimental,candes2013phaselift,candes2015phase_wf}.
The linear operator $\sensing$ depends on the phase-imaging modalities.
In the context of dense random models, the elements in $\sensing$ are i.i.d. sampled from a complex-valued Gaussian distribution.
In ptychography and Fourier ptychography, the forward matrix is a vertical concatenation of $\fourier \diag$ and $\fourier \diag \fourier$, respectively, where $\fourier$ stands for a Fourier matrix and $\diag$ a diagonal matrix.
Existing imaging modalities often leverage Fourier transforms and diagonal matrices, as these operations correspond to the ones performed by optical lenses and thin transmissive elements, respectively.
% The design of imaging modalities often leverages Fourier transforms and diagonal matrices due to the practical implementability of their physical counterparts, namely lenses and diffusers, respectively. When utilizing diffusers with random phase variations, multiple of them can be stacked interchangeably with the lenses to introduce more randomness into the imaging system progressively.

For the structured random model, we define the forward matrix as
\begin{equation} \label{def:sr}
	\sensing = \unders \left( \prod_{i=1}^N \trans \diag_i \right) \trans^s \overs,
\end{equation}
where $\diag_i \in \mathbb{C}^{p \times p}$ denotes a diagonal matrix with i.i.d. random diagonal elements of intermediate dimension $p = \max(m,n)$.
The amplitude and phase of the $r$th diagonal element $d^{(i)}_r = \rho^{(i)}_r \mathrm{e}^{\mathrm{j}\theta^{(i)}_r}$ of $\diag_i$ are independently sampled for $r = 0, \ldots, p-1$.

The matrix $\trans$ represents the discrete Fourier transform with complex entries
\begin{equation} \label{eq:dft}
	\trans = [f_{rc}] = \frac{1}{\sqrt{p}}\mathrm{e}^{\j \frac{2 \pi}{p} rc},
\end{equation}
where $r,c$ represent the row and column indices of the entry, ranging from 0 to $p-1$, respectively.

The undersampling $\unders = \begin{bmatrix} \mathbf{I}_m, \mathbf{0}\end{bmatrix} \in \complex^{m \times p}$ and oversampling $\overs = \begin{bmatrix} \mathbf{I}_n, \mathbf{0}\end{bmatrix}^\T \in \complex^{p \times n}$ matrices represent possible trimming and zero-padding operations.

We define the depth $L$ (i.e., the number of layers) of the model in Equation~\ref{def:sr} as $L=N+0.5s$, with $N \in \mathbb{N}$ and $s \in \{0,1\}$.
Each combination of $\trans \diag_i$ stands for one layer.
The optional Fourier transform $\trans^s$ contributes an additional 0.5 layer if $s=1$. When $s = 0$, $\trans^0 = \mathbf{I}$, meaning no extra depth is added.
Two important configurations are the \textbf{1.5-layer} and \textbf{2-layer} structured random transforms, with the respective transform sequence
\begin{align}
	\sensing_{1.5} & = \unders \trans \diag_1 \trans \overs,         \\
	\sensing_{2}   & = \unders \trans \diag_1 \trans \diag_2 \overs.
\end{align}

For greater generality, the Fourier transform in (\ref{def:sr}) can be replaced by another structured transform: a dense linear transform that can be evaluated in sub-$\mathcal{O}(n^2)$ complexity.
The intuition behind the structured random model is that its two key components emulate dense random matrices: diagonal matrices inject randomness, while structured transforms mix the different elements.

From a practical perspective, the structured random model is directly realizable with standard optical components. Each Fourier transform corresponds to a lens, and each random diagonal matrix to an optical diffuser. The resulting imaging system is an alternating cascade of diffusers and lenses, as illustrated in Figure~\ref{fig:optics}. While the one-layer model FD already corresponds to a coded-diffraction setup, our analysis in Section II-B shows that two structured layers are required to fully emulate a dense i.i.d. random matrix, which makes such a system meaningfully distinct from existing modalities. A practical limitation is that the diffuser phase patterns must be calibrated before reconstruction, which itself amounts to solving a phase retrieval problem per diffuser; this calibration only needs to be performed once per system and can leverage the same algorithms studied in this paper.

% The structured random model can be realized practically. We illustrate in Figure~\ref{fig:optics} the optical systems for both 1.5-layer and 2-layer configurations.
% Each Fourier transform corresponds to a lens, and each diagonal matrix to a diffuser.
% This hardware setup is simple to build and to calibrate in practice.

We shall use random diagonal matrices $\diag_i$ with: $\mathbb{E} \left[ \left(\rho^{(i)}_r\right)^2 \right] = 1$ for asymptotic energy conservation and $\mathbb{E}\left[d^{(i)}_r\right] = 0$.

\begin{figure}[tbp!]
	\centering
	\begin{subfigure}{1.0\linewidth}
		\centering
		\includegraphics[width=0.8\textwidth]{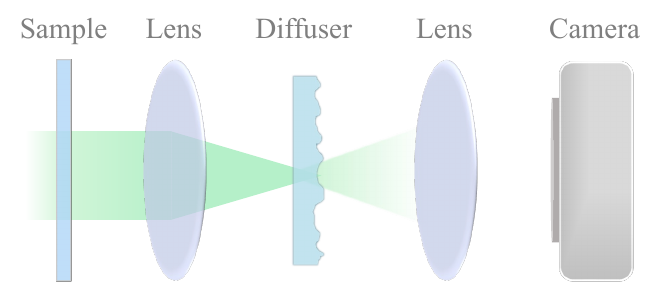}
		\caption{Optical implementation of a 1.5-layer model}
		\vspace{15pt}
	\end{subfigure}
	\begin{subfigure}{1.0\linewidth}
		\centering
		\includegraphics[width=1.0\textwidth]{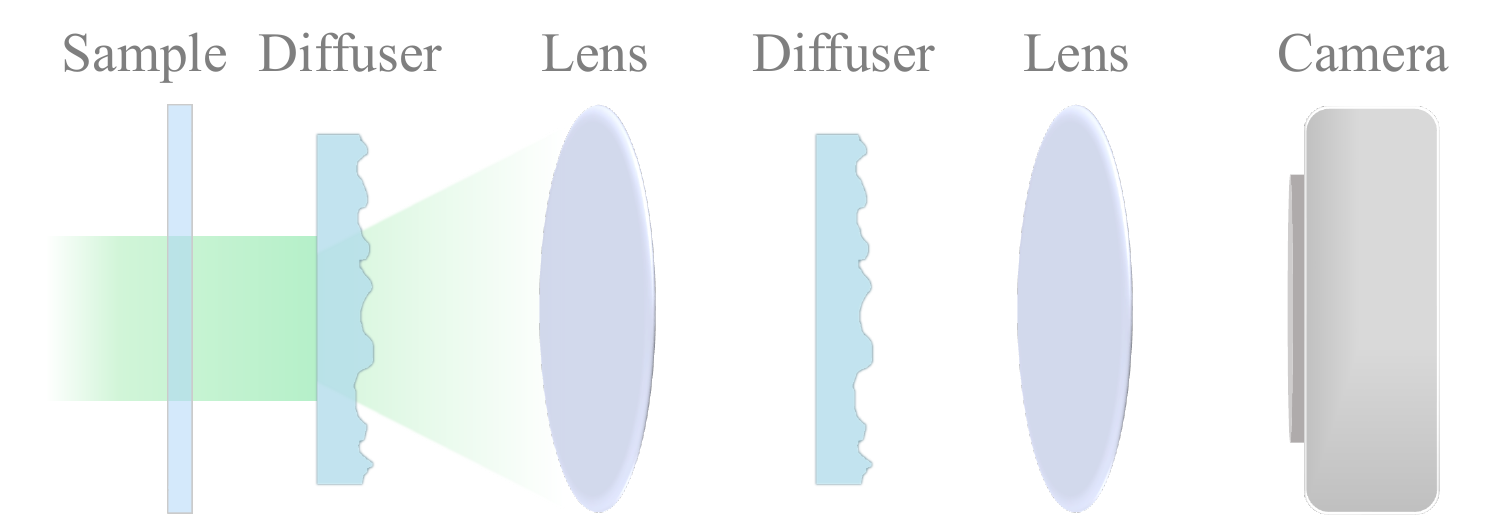}
		\caption{Optical implementation of a 2-layer model}
	\end{subfigure}
	\caption{Diagrams of possible optical implementations of structured random models with 1.5 and 2 layers, with each structured transform corresponding to a lens and each random diagonal matrix to a diffuser.}
	\label{fig:optics}
\end{figure}

\subsection{Covariance Analysis}
While the design in (\ref{def:sr}) permits an arbitrary model depth, the determination of the optimal depth remains crucial from a practical perspective due to the increased complexity inherent with deeper models. This motivates two fundamental questions: (i) whether the structured random model can replicate the statistical properties of a dense random model; and (ii) how many layers are required to obtain such equivalence. To analyze the randomness of the model, we study the covariance between the elements of the sensing matrix $\sensing$ as a function of depth.
% This second-order metric captures the lowest-order dependence between random variables. Moreover, for Gaussian random variables, the multivariate distribution is completely characterized by its first- and second-order moments.

We show in the following that the covariances exhibit distinct patterns. For the computation, we denote the elements in the structured transform and random diagonal matrix as $\trans= [f_{rc}] \in \complex^{p \times p}$ and $\diag = \operatorname{diag}([d_r]) \in \complex^{p \times p}$. All results are derived for the case of a square matrix and extend directly to the undersampled and oversampled cases.

We first recall the covariance matrix for the dense i.i.d. matrix, then derive it for structured random models with different number of layers.
% We now formulate the covariance pattern of a dense i.i.d. matrix and consider the three particular cases that correspond to 1 layer, 1.5 layers, and 2 or more layers.
\begin{proposition}
	For a dense matrix containing independent elements with variance $1/p$, the covariance between two elements at position $(r_1,c_1)$ and $(r_2,c_2)$ is
	\begin{align}
		\operatorname{Cov}(a_{r_1 c_1}, a_{r_2 c_2})
		= \begin{cases}
			  \frac{1}{p}, & \text{if $r_1=r_2$ and $c_1=c_2$} \\
			  0,           & \text{otherwise}.
		  \end{cases}
	\end{align}
\end{proposition}
% The position notation $(m,n)$ and $(k,l)$ for the two elements will also be used in later covariance computations.
\subsubsection{One Layer}
\begin{proposition} \label{theorem:1-layer}
	For a 1-layer model, the covariance between two elements in the forward matrix is given by
	\begin{equation}
		\operatorname{Cov}(a_{r_1 c_1},a_{r_2 c_2}) = \begin{cases}
			|d_{c_2}|^2 f_{r_1 c_1}\overline{f_{r_2 c_2}}, & \text{if } c_1 = c_2, \\
			0,                                             & \text{otherwise}.     %\text{if } c_1 \neq c_2.
		\end{cases}
	\end{equation}
\end{proposition}
\begin{proof}
	The element at row $r$ and column $c$ of the forward matrix of a 1-layer model is:
	\begin{equation}
		% \left[ \trans \diag \right]_{r,c} = d_c f_{rc},
		a_{rc} = d_c f_{rc},
	\end{equation}
	The covariance between two elements can thus be computed as
	\begin{align*}
		 & \operatorname{Cov}(d_{c_1} f_{r_1 c_1},d_{c_2} f_{r_2 c_2})                                                                                                \\
		 & \qquad = \mathbb{E}\left[(d_{c_1} f_{r_1 c_1}- \mathbb{E} [d_{c_1} f_{r_1 c_1}])(\overline{d_{c_2} f_{r_2 c_2} - \mathbb{E} [d_{c_2} f_{r_2 c_2}])}\right] \\
		 & \qquad = \mathbb{E} \left[ d_{c_1} f_{r_1 c_1} \overline{d_{c_2} f_{r_2 c_2}}\right] \ann{as $\mathbb{E}[d_r] = 0$},                                       \\
	\end{align*}
	which directly yields the result by the independence of the diagonal elements $d_r$.
\end{proof}
Proposition~\ref{theorem:1-layer} shows that the resulting matrix exhibits nonzero covariances within each column. This implies that each row of the sensing matrix $\mathbf{a}_r^\ctrans$ is strongly correlated with the others.
Since each measurement in $\measure$ is computed as the modulus of the inner product between the signal and one row of $\sensing$, these covariances heavily reduce the effective span of the sensing vectors, thereby making the problem more difficult.

To verify the theoretical derivations, we sample 10000 structured random matrices of size $4 \times 4$ and compute their empirical covariances for all pairs between the 16 elements.
Figure~\ref{fig:covariance} illustrates the covariance matrices for different model depths. The index $i$ represents the element position by row index $r = \lfloor i / p \rfloor$, where $\lfloor \cdot \rfloor$ is the floor operation, and column index $c = i \mod p$. It can be seen that for 1 layer, every pair of elements from the same column have a nonzero covariance.
\begin{figure}[t!]
	\centering
	\includegraphics[width=0.5\textwidth]{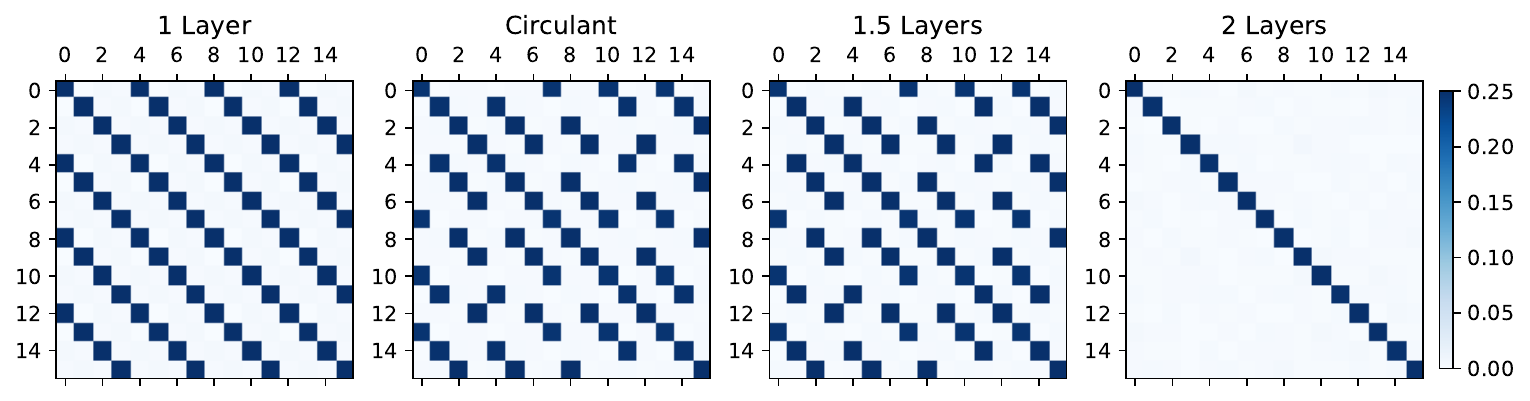}
	\caption{Modulus of the covariance matrix elements of ($4 \times 4$) structured random models, shown for different numbers of layers. Increasing depth suppresses off-diagonal elements, leading to uncorrelated entries at 2 layers.
	}
	\label{fig:covariance}
\end{figure}

\subsubsection{One and a Half Layers}\label{sec:cov-1p5}
\begin{proposition} \label{theorem:1p5}
	For a 1.5-layer model, the forward matrix is left-circulant with asymptotically Gaussian-distributed elements with zero mean and variance $1/p$. The covariance matrix is given by
	\begin{equation}
		\operatorname{Cov}(a_{r_1 c_1},a_{r_2 c_2})
		= \begin{cases}
			\frac{1}{p}, & \text{if } ((r_1+c_1)-(r_2+c_2)) \operatorname{mod} p = 0 \\
			0,           & \text{otherwise}.
		\end{cases}
	\end{equation}
\end{proposition}
\begin{proof}
	For 1.5 layers, the element at row $r$ and column $c$ of the forward matrix is
	\begin{equation} \label{eq:1p5}
		a_{rc} = \sum_{k=0}^{p-1} d_k f_{rk} f_{kc}.
	\end{equation}
	It is a summation of pseudosamples $d_k f_{rk} f_{kc}$, which consist of random samples $d_k$ and structured modulations $f_{rk} f_{kc}$. By the central limit theorem (CLT), this summation asymptotically follows a Gaussian distribution for large values of $p$. The mean of $a_{rc}$ is zero as each sample $d_k$ has zero mean.
	The corresponding covariance is given by
	\begin{align*}
		 & \operatorname{Cov}(\sum_{k_1=0}^{p-1} d_{k_1} f_{r_1 k_1} f_{k_1 c_1},\sum_{k_2=0}^{p-1} d_{k_2} f_{r_2 k_2} f_{k_2 c_2})                                                     \\
		 & \qquad = \mathbb{E} \left[ \sum_{k_1=0}^{p-1} d_{k_1} f_{r_1 k_1} f_{k_1 c_1} \overline{\sum_{k_2=0}^{p-1} d_{k_2} f_{r_2 k_2} f_{k_2 c_2}} \right] \ann{$\mathbb{E}[d_r]=0$} \\
		 & \qquad = \mathbb{E} \left[ \sum_{k_1=0}^{p-1} |d_{k_1}|^2 f_{r_1 k_1} f_{k_1 c_1} \overline{f_{r_2 k_1} f_{k_1 c_2}} \right] \ann{$d_r$ is i.i.d.}                            \\
		 & \qquad = \sum_{k_1=0}^{p-1} f_{r_1 k_1} f_{k_1 c_1} \overline{f_{r_2 k_1} f_{k_1 c_2}} \ann{$\mathbb{E}\left[ |d_r|^2 \right] = 1$}                                           \\
		 & \qquad = \sum_{k_1=0}^{p-1} \frac{1}{p^2} \mathrm{e}^{\j \frac{2\pi}{p} k_1(r_1+c_1-r_2-c_2)},
	\end{align*}
	where the last simplification results from (\ref{eq:dft}).
	The covariance pattern above shows that the forward matrix is left-circulant, with uncorrelated elements in each row.
\end{proof}
Figure~\ref{fig:covariance} illustrates the covariance profile of a 1.5-layer model, which indeed coincides with that of a left-circulant matrix.
% Notably, this structural dependency distinguishes the formulation from a standard i.i.d. dense random model.

%Compared to a 1-layer, the covariances in the same column of the sensing matrix are destroyed, but new covariances between different columns are introduced. \textbf{I would move the last sentence of the demo, and avoid using "pattern" and "destroyed". Add for clarity that this is not the same as an iid random model.}
% As the forward matrix $\sensing = \trans \diag \trans$ is left-circulant, the forward equation can be rewritten as:
% \begin{equation*}
%     \measure = |\operatorname{circ}(\mathbf{a}_1^\ctrans)\signal|^2 = |\operatorname{circ}(\signal)\mathbf{a}_1^\ctrans|^2,
% \end{equation*}
% where $\operatorname{circ}(\cdot)$ is the left-circulant matrix constructor given a vector and $\mathbf{a}_1^\ctrans$ is the first row of $\sensing$. This suggests that the forward process can be treated as sensing the random vector $\mathbf{a}_1^\ctrans$ from the circulant matrix $\operatorname{circ}(\signal)$.
\subsubsection{Two Layers} \label{sec:2layers}
\begin{proposition} \label{theorem:2-layer}
	For a 2-layer model, the elements in the forward matrix are distributed according to an asymptotically Gaussian distribution with zero mean and covariance matrix
	% For a model of depth $L \geq 2$ layers, the elements in the forward matrix emulate uncorrelated samples from an asymptotically Gaussian distribution with zero mean and variance $1/p$. The covariance matrix is given by
	\begin{equation}
		\label{eq:theorem-2-layers}
		\operatorname{Cov}(a_{r_1 c_1}, a_{r_2 c_2})
		= \begin{cases}
			\frac{1}{p}, & \text{if $r_1=r_2$ and $c_1=c_2$} \\
			0,           & \text{otherwise}.
		\end{cases}
	\end{equation}
\end{proposition}
\begin{proof}
	For $L = 2$ layers, the $r$th row and $c$th column element of the forward matrix is:
	\begin{equation}
		% \left[\trans \diag \trans \diag \right]_{r,c} = d'_c\sum_{k=0}^{p-1} d_k f_{rk} f_{kc}.
		a_{rc} = d'_c\sum_{k=0}^{p-1} d_k f_{rk} f_{kc}.
	\end{equation}
	Each element of the forward matrix is then a multiplication of an element of the 1.5-layer model with an extra random sample $d'_c$. The element $a_{rc}$ asymptotically follows a Gaussian distribution according to the CLT, with zero mean since the diagonal elements $d'_c$ themselves have zero mean.
	The covariance is computed as
	\begin{align*}
		 & \operatorname{Cov}(d_{c_1}'\sum_{k_1=0}^{p-1} d_{k_1} f_{r_1 k_1} f_{k_1 c_1},d_{c_2}'\sum_{k_2=0}^{p-1} d_{k_2} f_{r_2 k_2} f_{k_2 c_2})                                              \\
		 & \qquad = \mathbb{E} \left[ d_{c_1}'\overline{d_{c_2}'}\sum_{k_1=0}^{p-1} |d_{k_1}|^2 f_{r_1 k_1} f_{k_1 c_1} \overline{f_{r_2 k_1} f_{k_1 c_2}} \right] \ann{$d_r$ is i.i.d.}          \\
		 & \qquad = \mathbb{E}\left[d_{c_1}'\overline{d_{c_2}'}\right]\sum_{k_1=0}^{p-1} f_{r_1 k_1} f_{k_1 c_1} \overline{f_{r_2 k_1} f_{k_1 c_2}}  \ann{$\mathbb{E}\left[ |d_r|^2 \right] = 1$} \\
		 & \qquad = \mathbb{E}\left[d_{c_1}'\overline{d_{c_2}'}\right] \sum_{k_1=0}^{p-1} \frac{1}{p^2} e^{\j \frac{2\pi}{p} k_1(r_1+c_1-r_2-c_2)} \ann{Equation (\ref{eq:dft})}                  \\
		 & \qquad = \begin{cases}
			            \frac{1}{p}, & \text{if } ((r_1+c_1)-(r_2+c_2)) \operatorname{mod} p = 0 \; \text{and} \; c_1=c_2 \\
			            0,           & \text{otherwise},
		            \end{cases}
		% = &\begin{cases}
		%     \frac{1}{p} & \text{if $r_1=r_2$ and $c_1=c_2$}, \\
		%     0 & \text{otherwise}.
		%     \end{cases}
	\end{align*}
	which gives~\eqref{eq:theorem-2-layers} as $((r_1+c_1)-(r_2+c_2)) \operatorname{mod} p = 0$ implies that $r_1 + c_1 = r_2 + c_2$.
\end{proof}

% We then assume that Proposition~\ref{theorem:2-layer} holds true for $L \geq 2$.
% For depth $L+0.5$, the corresponding model is $\sensing_{L+0.5} = \sensing_L \fourier = \sensing_{L-0.5} \diag \fourier$, with each element $a^{(L+0.5)}_{rc} = \sum_{k=0}^{p-1} d_k f_{kc} a_{rk}^{(L-0.5)}$.
% Given that the diagonal elements $d_r^{(L)}$ across different layers are independent, $a^{(L+0.5)}_{rc}$ is a summation of i.i.d. samples $d_k$ and perturbations $f_{kc} a_{rk}^{(L-0.5)}$.
% The distribution of $a^{(L+0.5)}_{rc}$ conditioned on $a_{rk}^{(L-0.5)}$ is thus asymptotically Gaussian. \textbf{can only show the distribution of $a^{(L+0.5)}_{rc}$ given $a_{rk}^{(L-0.5)}$ is asymptotically Gaussian}
% As established by the induction hypothesis, the elements $a_{rk}^{(L)}$ for $k=0,\ldots,p-1$ are uncorrelated and asymptotically jointly Gaussian.
% In the limit, since uncorrelated jointly Gaussian variables are statistically independent, the terms in the summation become independent. We can thus invoke the standard CLT to verify that the summation $a^{(L+0.5)}_{rc} = \sum_{k=0}^{p-1} a_{rk}^{(L)} f_{kc}$ follows an asymptotic Gaussian distribution with zero mean.
The extra sample $d'_c$ in the 2-layer model further decouples the elements in different columns, eliminating the remaining covariances in a 1.5-layer model. As verified in Figure~\ref{fig:covariance}, the covariance matrix of a 2-layer model is diagonal, matching the one of the i.i.d. random case.
% While this does not imply that the elements are strictly independent, their first and second moments exactly match those of an i.i.d. random model.

\subsubsection{More than Two Layers}
\begin{proposition} \label{theorem:2+-layer}
	For a structured random model with $L > 2$ layers, the elements in the forward matrix have zero mean and their covariance matrix is given by
	\begin{equation}
		\label{eq:theorem-2+-layers}
		\operatorname{Cov}(a_{r_1 c_1}, a_{r_2 c_2})
		= \begin{cases}
			\frac{1}{p}, & \text{if $r_1=r_2$ and $c_1=c_2$} \\
			0,           & \text{otherwise}.
		\end{cases}
	\end{equation}
\end{proposition}

\begin{proof}
	We prove Proposition~\ref{theorem:2+-layer} by induction. First, Proposition~\ref{theorem:2+-layer} holds for $L = 2$ directly by Proposition~\ref{theorem:2-layer}.

	Now assume Proposition~\ref{theorem:2+-layer} holds for layer $L \geq 2, L \in \N$. For depth $L+0.5$, we write $\sensing_{L+0.5} = \sensing_{L}\fourier$, giving each elements as
	\begin{equation}
		a^{(L+0.5)}_{rc}=\sum_{k=0}^{p-1} a^{(L)}_{rk}f_{kc}.
	\end{equation}
	The mean of $a^{(L+0.5)}_{rc}$ is zero as $\mathbb{E}[a_{rc}^{(L)}]=0$.
	The covariance is computed as
	\begin{align*}
		 & \operatorname{Cov}{(\sum_{k_1=0}^{p-1} a_{r_1 k_1}^{(L)} f_{k_1 c_1},\sum_{k_2=0}^{p-1} a_{r_2 k_2}^{(L)} f_{k_2 c_2})}                                                            \\
		 & \qquad = \mathbb{E} \left[ \sum_{k_1=0}^{p-1} a_{r_1 k_1}^{(L)} f_{k_1 c_1} \overline{\sum_{k_2=0}^{p-1} a_{r_2 k_2}^{(L)} f_{k_2 c_2}} \right] \ann{$\mathbb{E}[a_{rc}^{(L)}]=0$} \\
		 & \qquad = \begin{cases}
			            \sum_{k_1=0}^{p-1} \mathbb{E}\left[|a_{r_1 k_1}^{(L)}|^2\right] f_{k_1 c_1} \overline{f_{k_1 c_2}}, & \text{if } r_1 = r_2 \\
			            0,                                                                                                  & \text{otherwise}
		            \end{cases} \ann{$a_{rc}^{(L)}$ uncorrelated}                                               \\
		 & \qquad = \begin{cases}
			            \frac{1}{p}, & \text{if } r_1 = r_2 \; \text{and} \; c_1 = c_2 \\
			            0,           & \text{otherwise}.
		            \end{cases}
	\end{align*}
	The last equality follows from the fact that $a_{rc}^{(L)}$ has zero mean and variance $1/p$, and that the Fourier matrix $\trans$ is unitary.

	For depth $L+1$, the model is $\sensing_{L+1} = \sensing_{L+0.5} \diag$, each element $a^{(L+1)}_{rc} = a^{(L+0.5)}_{rc} d_c$ also has zero mean. The covariance is computed as
	\begin{align*}
		 & \operatorname{Cov}{(a_{r_1 c_1}^{(L+0.5)} d_{c_1}, a_{r_2 c_2}^{(L+0.5)} d_{c_2})}                                                                          \\
		 & \qquad = \mathbb{E} \left[ a_{r_1 c_1}^{(L+0.5)} \overline{a_{r_2 c_2}^{(L+0.5)}} d_{c_1} \overline{d_{c_2}} \right] \ann{$\mathbb{E}[a_{rc}^{(L+0.5)}]=0$} \\
		 & \qquad = \begin{cases}
			            \frac{1}{p} & \text{if } r_1 = r_2 \; \text{and} \; c_1 = c_2, \\
			            0           & \text{otherwise}.
		            \end{cases}
	\end{align*}
	We thus conclude by induction that Proposition~\ref{theorem:2+-layer} holds for all $L > 2$, including integers and half integers.
\end{proof}

\section{Methods} \label{sec:methods}
In this section, we describe the algorithms used to obtain the reconstruction from the measurements.
We focus on gradient descent and spectral methods as they offer a good balance between performance and efficiency. Furthermore, we describe the experimental setup for the evaluations, including the signal and model configurations as well as algorithm and evaluation metric choices.

\subsection{Gradient Descent}
Gradient descent (GD) is a general-purpose iterative method for nonlinear optimization. The gradient of the loss function can often be derived in closed form. For instance, the gradient of the standard $\ell_2$ loss
\begin{equation}
	\mathcal{L}(\hat{\signal}) = \| |\sensing\hat{\signal}|^2 - \measure \|^2,
\end{equation}
is given by
\begin{equation}
	\label{eq:l2_gradient}
	\mbox{{\boldmath{$\nabla$}}}_{\hat{\signal}} \mathcal{L} = - 2 \sensing^\ctrans\operatorname{diag}(\sensing \hat{\signal})(\measure-|\sensing \hat{\signal}|^2),
\end{equation}
where $\operatorname{diag}(\cdot)$ denotes the diagonal matrix whose diagonal is the argument vector.
The evaluation of the gradient involves matrix multiplications with both the forward matrix $\sensing$ and its adjoint $\sensing^\ctrans$. In general, the complexity is determined by the structure of the forward model, while some momentum techniques, such as the Nesterov acceleration~\cite{nesterov1983method}, utilize the historical information collected throughout the iterations to further accelerate the convergence.

%Besides classical acceleration techniques, including utilizing the Hessian of the loss function and adaptive step sizes, machine learning regularization has also attracted increasing interest in recent research. It originates from the idea of adding prior beliefs on the reconstruction through classical regularizers such as total variation or Tikhonov \tocite. These additional priors are typically satisfied through an additional proximal step during two gradient steps on the original data fidelity loss. Instead of hand-designing the prior with general-purpose guarantees, machine learning regularizers try to encode the prior implicitly in a neural network and call it during the proximal step. These plug-and-play (PnP) \tocite regularizers provide additional information on the reconstruction direction during gradient descent, and can be trained on denoising once and reused for any other modalities. Benefiting from seeing the image data distribution, PnP regularizers can adaptively denoise according to the input image structure compared to classical ones, thus providing a better performance.

Despite an abundance of available algorithms, GD lacks convergence guarantees for phase retrieval due to the theoretical difficulty introduced by the nonlinearity. To improve convergence in practice, we typically use spectral methods as an initialization of GD to improve its convergence.
\subsection{Spectral Methods}
Spectral methods (SM) stand for a class of reconstruction algorithms specific to the quadratic nature of phase retrieval. They aim to find the leading eigenvector of the covariance matrix~\cite{ma2021spectral,luo2019optimal,maillard2021construction}
\begin{equation}
	\label{eq:covariance_matrix}
	\covar_\mathcal{T} = \frac{1}{m} \sum_{i=1}^m \mathcal{T}(y_i) \mathbf{a}_i \mathbf{a}_i^\ctrans =
	\sensing^\ctrans \text{diag}\left({\frac{\mathcal{T}(\measure)}{m}}\right) \sensing,
\end{equation}
where $y_i$ is the $i$th element of the measurements $\measure$, $\mathbf{a}_i^\ctrans$ the $i$th row of the sensing matrix $\sensing$, and $\mathcal{T}$ an increasing preprocessing function used to assign weights to the outer products $\mathbf{a}_i \mathbf{a}_i^\ctrans$.
The intuition is that each measurement $y_i$ represents the magnitude of the inner product between $\signal$ and $\mathbf{a}_i$, and reflects their covariance.
A larger $y_i$ suggests that $\mathbf{a}_i$ is more closely aligned with $\signal$ and should be assigned a heavier weight.
Among all preprocessing functions, the optimal one for Gaussian i.i.d. random phase retrieval has been derived as $\mathcal{T}(y) = \text{max}(1 - 1/y, -b)$, where $b > 0$ defines a lower bound for negative eigenvalues~\cite{luo2019optimal}.
The covariance matrix $\mathbf{M}_\mathcal{T}$ is a second-order quantity that should provide a bias towards $\signal$ in its principal eigenvector.
Although SM can be applied to any forward-matrix configuration, theoretical and numerical studies have largely focused on the i.i.d. random case.

In practice, the principal eigenvector is obtained through power iterations.
Starting from a random initial guess $\hat{\signal}_0$, we iteratively compute
\begin{equation}
	\hat{\signal}_{t+1} = \mathbf{M}_\mathcal{T} \hat{\signal}_t + 2 b\hat{\signal}_t,
\end{equation}
where $t$ is the iteration index. The trailing term $2 b \hat{\signal}_t$ serves as regularization against negative eigenvalues. While GD can be trapped in local minima, SM is guaranteed to converge to the original signal.
% \subsubsection{Other algorithms}
% Other reconstruction algorithms have been proposed to solve the phase retrieval problem. Projection algorithms treat the measurements as constraints imposed on the signal and apply projection operators iteratively on the estimate to enforce these constraints. With convex relaxation, one optimizes on an auxiliary variable, typically of higher dimension, to construct a convex optimization program~\cite{candes2013phaselift}.
% In Bayesian algorithms such as Approximate Message Passing, one estimates a posterior distribution of the signal $\mathbb{P}(\signal \mid \measure, \sensing)$ given the model and measurements~\cite{mondelli2019fundamental,maillard2020phase,dong2023phase}.
% Efficient algorithms are required to fully utilize the computational efficiency of structured random models.
\subsection{Experimental Setup}
\subsubsection{Model Setup} \label{sec:model-setup}
We evaluate the performance of the 2-layer structured random model $\sensing = \unders \trans \diag_1 \trans \diag_2 \overs$, with $\fourier$ applied by the fast Fourier transform (FFT).
The choice of the number of layers will be further investigated in Section~\ref{sec:depth}.
% The magnitudes of the diagonals follow a Marchenko configuration, i.e., the magnitudes of diagonal elements in $\diag_2$ follow the square root of $\mathcal{MP}_\alpha$, and those of $\diag_1$ are ones. This is to emulate the spectrum of an i.i.d. random matrix.
The phase of the diagonal elements of $\diag_1$ and $\diag_2$ are sampled uniformly between $0$ and $2 \pi$, and the diagonal elements of $\diag_1$ all have unit magnitude.
To emulate i.i.d. random matrices, we sample the squared amplitudes of the diagonal elements of $\diag_2$ from the Marchenko-Pastur distribution, which describes the asymptotic eigenvalue distribution of $\mathbf{\mathbf{X}^\ctrans \mathbf{X}}$, where $\mathbf{X} \in \mathbb{C}^{m \times n}$ is an i.i.d. sampled Gaussian matrix.
With $\lambda_{\pm} = \sigma^2(1 \pm \sqrt{1/\alpha})^2$, where $\sigma^2$ is the variance of the sampling Gaussian distribution of $\mathbf{X}$. This distribution is given by
\begin{equation}
	\mathcal{MP}_\alpha(\lambda) =
	\max(0, 1-\alpha) \delta(\lambda) + \mathbf{1}_{[\lambda_-, \lambda_+]}(\lambda) \nu(\lambda),
\end{equation}
in which $\delta$ is a Dirac distribution and $\mathbf{1}_{[\lambda_-, \lambda_+]}$ is the indicator function of $[\lambda_-, \lambda_+]$. The bulk distribution is
\begin{equation}
	\nu(\lambda) = \frac{\alpha \sqrt{(\lambda_+ - \lambda)(\lambda - \lambda_-)}}{2 \pi \sigma^2 \lambda}.
\end{equation}
\subsubsection{Signal Configuration}
We use the standard Shepp-Logan image of ($64 \times 64$) pixels as the phase information of the signal to be reconstructed. The grayscale pixel values are mapped linearly from $[0, 1)$ to a phase in $[-\pi, \pi)$. The magnitudes of the phase signal are set to be constant and equal to 1. In later sections, we consistently use the Shepp-Logan image as the phase image, unless otherwise specified, while we keep the image size.
\subsubsection{Algorithms Configuration}
We use GD and SM, either separately or together, with SM returning an initial guess for the GD iterations. Further, the GD is combined with Nesterov acceleration~\cite{nesterov1983method} to achieve a lower full recovery threshold.
A comparison on different algorithms can be found in Section~\ref{sec:algorithm-comparison}.
We fix the hyperparameters as follows:
for SM, the preprocessing function is $\mathcal{T}(y) = \text{max}(1 - 1/y, - b)$ with regularization parameter $b = 10$;
the mean of $\mathbf{y}$ is renormalized to be 1;
we also set a maximum of 5,000 power iterations with early stopping upon convergence.
For GD, the maximal number of iterations is 10,000 with early stopping, and the stepsize is chosen as $2/\overline{\operatorname{Lip}_{\mathcal{L}}}$, where $\overline{\operatorname{Lip}_{\mathcal{L}}}$ is an estimated upper bound of the local Lipschitz constant of the loss function computed at the initial guess.
We let the numbers of iterations to ensure the convergence of each algorithm.
% Although the amplitude loss $\mathcal{L}(\hat{\signal}) = \| |\sensing\hat{\signal}| - \sqrt{\measure} \|^2$ is known to provide a more robust performance than the vanilla $\ell_2$ loss in practice~\cite{yeh2015experimental}, it provides the same performance for the noiseless setting in our case.
\subsubsection{Evaluation Metric}
Since the phase-retrieval problem is invariant to a global phase factor, in the sense that, both $\signal$ and $\signal \mathrm{e}^{\j \phi}$ will generate identical measurements, we use the cosine similarity metric to benchmark reconstruction accuracy. It is given by
\begin{equation}
	\label{eq:cosine_similarity}
	\cos(\truth, \guess) = \frac{\left | \truth^\ctrans \guess \right |}{\| \truth \| \cdot \| \guess \|},
\end{equation}
where $\truth$ is the ground-truth signal and $\guess$ the reconstructed signal.
This metric computes the alignment between two signals as a scalar in the range $[0,1]$, where 0 indicates no alignment and 1 represents a perfect reconstruction.
\subsubsection{Estimation of Signal Norms} \label{sec:norm-estimation}
As the principal eigenvector returned by SM always has unit norm, we only acquire the direction of the signal.
However, the quality of the initial guess of the norm is important for the subsequent GD iterations.
We estimate the $\ell_2$-norm of $\signal$ using the statistics of $\measure$ and $\sensing$ as
\begin{equation}
	\|\signal\|_2 \approx \sqrt{\frac{\mu_\measure}{\sigma_\sensing^2}},
\end{equation}
where $\mu_\measure$ and $\sigma_\sensing^2$ are the sample mean of $\measure$ and sample element-wise variance of $\sensing$, respectively.

A detailed derivation is as follows. The $i$th entry of $\measure$ can be written as
\begin{equation}
	y_i = \signal ^\ctrans \mathbf{a}_i \mathbf{a}_i^\ctrans \signal.
\end{equation}
Given that the rows $\mathbf{a}_i^\ctrans$ are independent, the measurements $\measure$ can be treated as independent random variables. Taking the expectation of $y_i$ yields that
\begin{align*}
	\mathbb{E}[y_i] & = \mathbb{E}[\signal ^\ctrans \mathbf{a}_i \mathbf{a}_i^\ctrans \signal]                                                       \\
	                & = \mathrm{Var}[\mathbf{a}_i^\ctrans \signal] + \left|\mathbb{E}[\mathbf{a}_i^\ctrans \signal]\right|^2 \ann{variance property} \\
	% &= \mathrm{Var}[\mathbf{a}_i^\ctrans \signal] + \left|\sum_{j=1}^n\mathbb{E}[a_{ij}^\ctrans x_j]\right|^2 \ann{$a_{ij}$ are independent} \\
	% &= \mathrm{Var}[\mathbf{a}_i^\ctrans \signal] + \left|\sum_{j=1}^n\mathbb{E}[a_{ij}^\ctrans] x_j\right|^2 \ann{$x_j$ is not random} \\
	                & = \mathrm{Var}[\mathbf{a}_i^\ctrans \signal] \ann{$a_{ij}$ with zero mean}                                                     \\
	% &= \signal^\ctrans \mathrm{Cov}(\mathbf{a}_i^\ctrans) \signal \\
	% &= \mathrm{Var}\left[\sum_{j=i}^n a_{ij}^\ctrans x_j\right] \\
	% &= \mathrm{Var}[a_{ij}] \signal^\ctrans \signal \ann{$a_{ij}$ are i.i.d.} \\
	                & = \mathrm{Var}[\sensing] \left\|\signal \right\|_2^2.
	%\ann{definition of $\mathrm{Var}[\sensing]$}
\end{align*}
% Therefore, we can estimate the 2-norm of $\signal$ using the sample mean of $\measure$ and the sample variance of $\sensing$.
\subsubsection{Code Availability}
The proposed forward models and the corresponding reconstruction algorithms are implemented using DeepInverse~\cite{tachella2025deepinverse}, an open-source computational imaging library. Our complete source code is made publicly available\footnote{\url{https://github.com/zhiyhucode/structured-random-phase-retrieval-2}}.

\section{Results} \label{sec:results}
In this section, we demonstrate the reconstruction equivalence between dense and structured random models, followed by a computational benchmark that validates the efficiency of the structured random model.
For a better understanding of the model structure, a comparison on different model depths shows that at least 2 layers are indeed necessary to reproduce the results obtained with i.i.d. random models.
We further demonstrate the practicality of the proposed model by showing its robustness to measurement noise, and its flexibility by introducing binary phase diffusers and other structured transforms.
Finally, we show that the performance of the model relies on the spectrum of the sensing matrix.

\subsection{Reconstruction Performance}

\begin{figure}[t!]
	\centering
	\begin{subfigure}[b]{0.95\linewidth}
		\centering
		\includegraphics[width=1.0\textwidth]{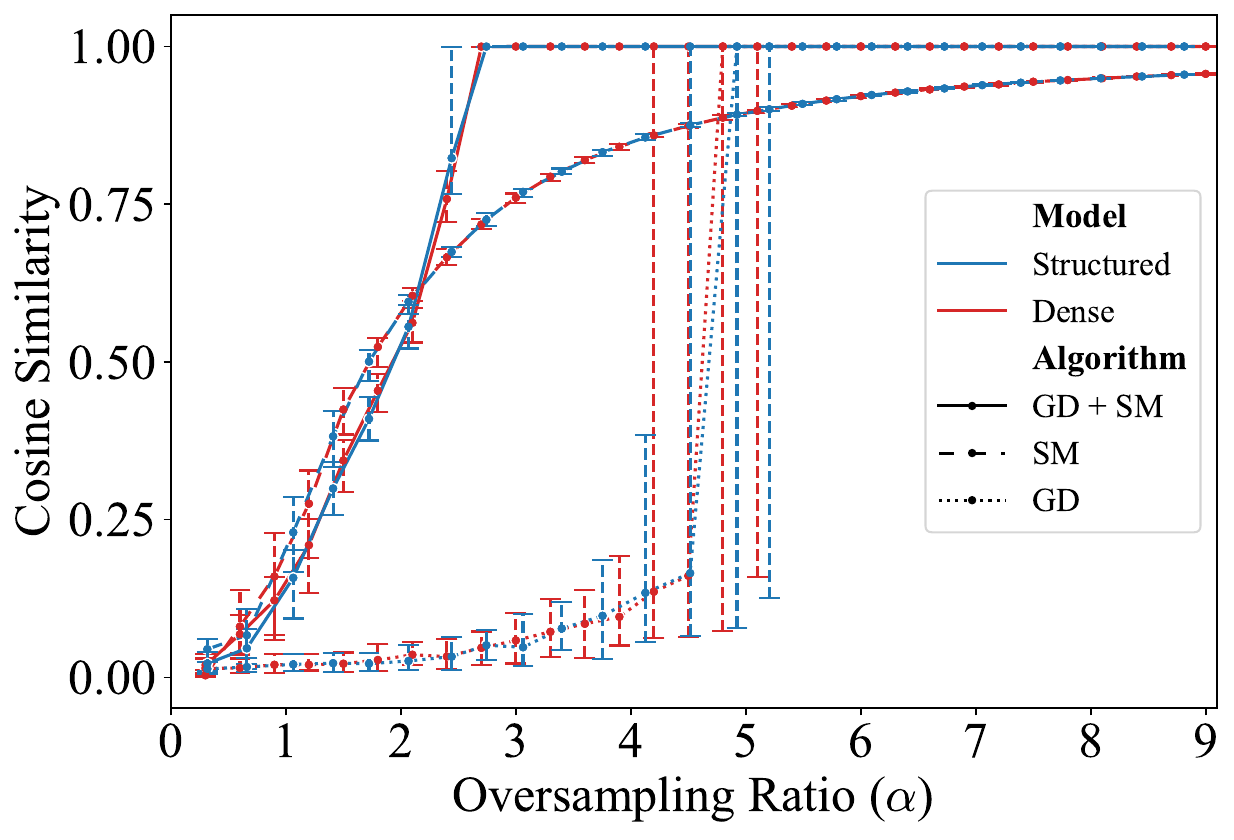}
		\caption{Reconstruction accuracy w.r.t the OR}
		\label{fig:algo}
		\vspace{15pt}
	\end{subfigure}
	\begin{subfigure}[b]{0.95\linewidth}
		\centering
		\includegraphics[width=1.0\textwidth]{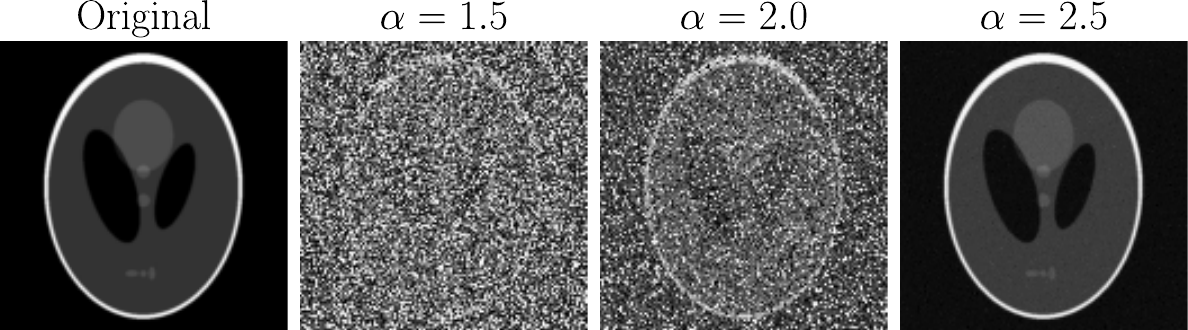}
		\caption{Reconstruction visualization using GD $+$ SM}
		\label{fig:viz}
	\end{subfigure}
	\caption{Reconstruction Comparison. Structured random models obtain the same accuracy as dense random models for each benchmarked algorithm, achieving a perfect recovery at an OR of 2.5. Nesterov accelerated gradient algorithm is used for the GD optimization.}
	\label{fig:main}
\end{figure}

We first evaluate the reconstruction performance with both dense and structured random models. Figure~\ref{fig:main} illustrates (a) the reconstruction performance w.r.t. the OR and provides (b) a few reconstructed images at several key ORs for structured random models.
The error bars represent the 10\%, 50\%, and 90\% percentiles among 50 runs, respectively. We show in Figure~\ref{fig:algo} that structured random models enable a robust reconstruction as good as in the case of dense random models using the same algorithms, which offers an applicable alternative without any performance compromise.

Regarding the choice of algorithms, standalone SM provides a stable and satisfactory performance, while a subsequent run of GD further improves the results.
The perfect recovery retrieved at an OR of 2.5 is comparable to the 2.03 threshold by AMP, and below the typical ORs of 4 in classic phase-imaging modalities. Moreover, standalone GD with random initialization performs significantly worse than GD with SM initialization, thus, unfavorable in practice. This arises from the highly nonconvex loss landscape induced by the phase-retrieval problem, which causes GD to become trapped in local minima. SM mitigates this difficulty by providing a suitable initialization in the attraction basin of the solution.

As illustrated in Figure~\ref{fig:viz}, a reconstruction under the OR of 1.5 reveals only the general structure of the original object. One can capture more details by increasing the OR to 2.0 and, eventually, enabling a perfect recovery at an OR of 2.5.

\subsection{Computation Benchmark}
\begin{figure}
	\centering
	\includegraphics[width=0.95\linewidth]{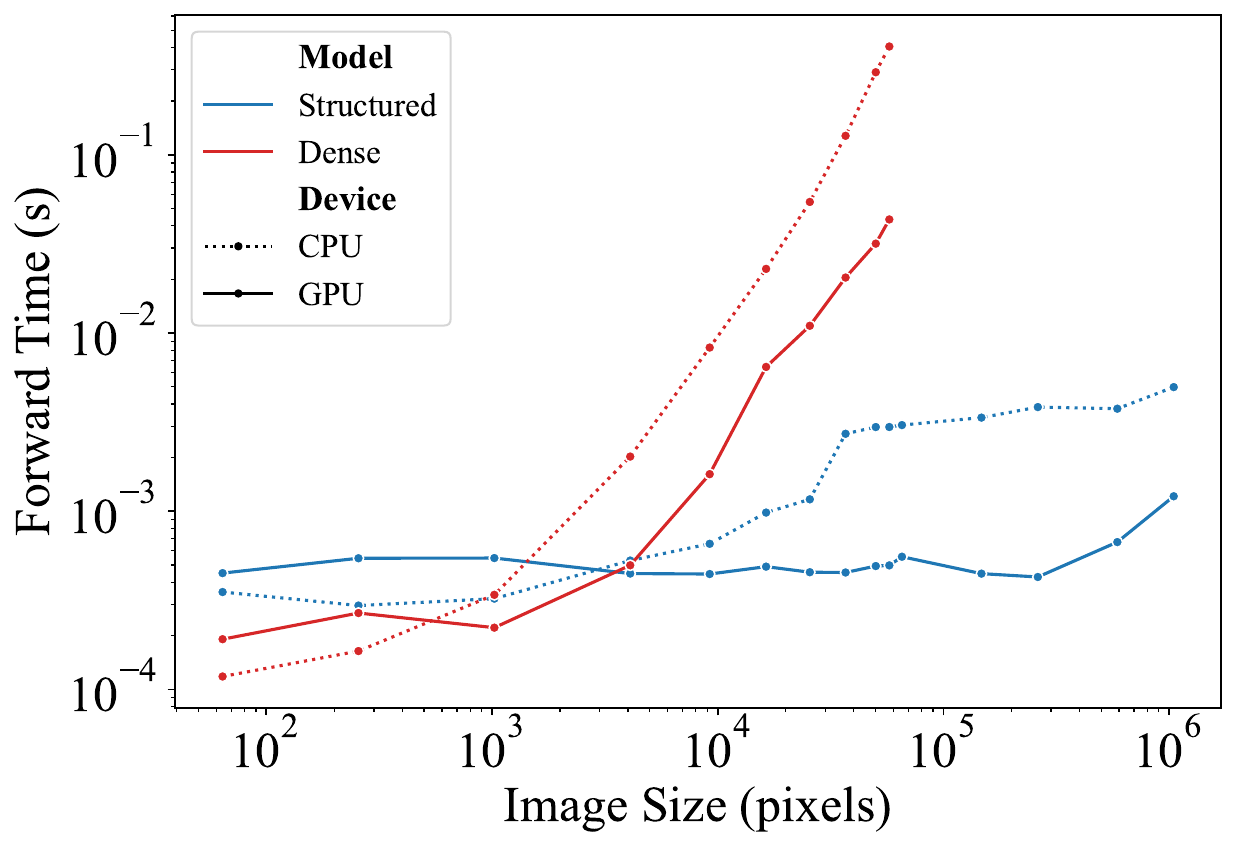}
	\caption{Time benchmark of the forward pass of dense and structured random models under an OR of 1. Structured random models reach nearly constant forward time on the GPU for moderate image sizes, whereas dense random models quickly become prohibitively slow.}
	\label{fig:time}
\end{figure}
To demonstrate the speed advantage of structured random models, we benchmark the forward time (i.e., the time to evaluate $\measure = \left|\sensing \signal\right|^2$) for dense and structured random models w.r.t. the size of the input.
The forward pass is evaluated during each iteration of the reconstruction algorithms and is the determining factor for the computational complexity.

We report in Figure~\ref{fig:time} the average over 100 runs of the forward time of both models for an OR of 1 on CPU Intel Xeon Gold 6240 and GPU NVIDIA Tesla V100 SXM2.
Due to GPU parallelization, the structured random model reaches a nearly constant forward time for moderate image sizes (fewer than $10^6$ pixels, which corresponds to a (1,024$\times$1,024) image). In comparison, the forward time of a dense random model resembles a quadratic curve, which matches the theoretical prediction and becomes excessively large already for small image sizes with $10^4$ pixels. For CPU, the forward time of the structured model remains manageable for small image sizes with fewer than $10^5$ pixels, whereas the dense random models are prohibitively slow.
For sizes smaller than $10^3$ pixels, overheads make the structured models a few times slower than the dense ones, although both computational costs are negligible at these modest sizes.

Our benchmark for dense random models is limited to image sizes smaller than $10^5$ pixels, in reason of the excessive memory requirements of dense matrices. Those scale quartically w.r.t. the number of pixels per edge. For instance, a dense random matrix would consume 17.18 GB for a ($256 \times 256$) image.
At the same image size, the storage needed for the proposed structured random model is around 1 MB, as only the diagonal matrices need to be stored. This is an improvement of a factor of 17,000.

\subsection{Depth} \label{sec:depth}
\begin{figure}[t!]
	\centering
	\includegraphics[width=0.95\linewidth]{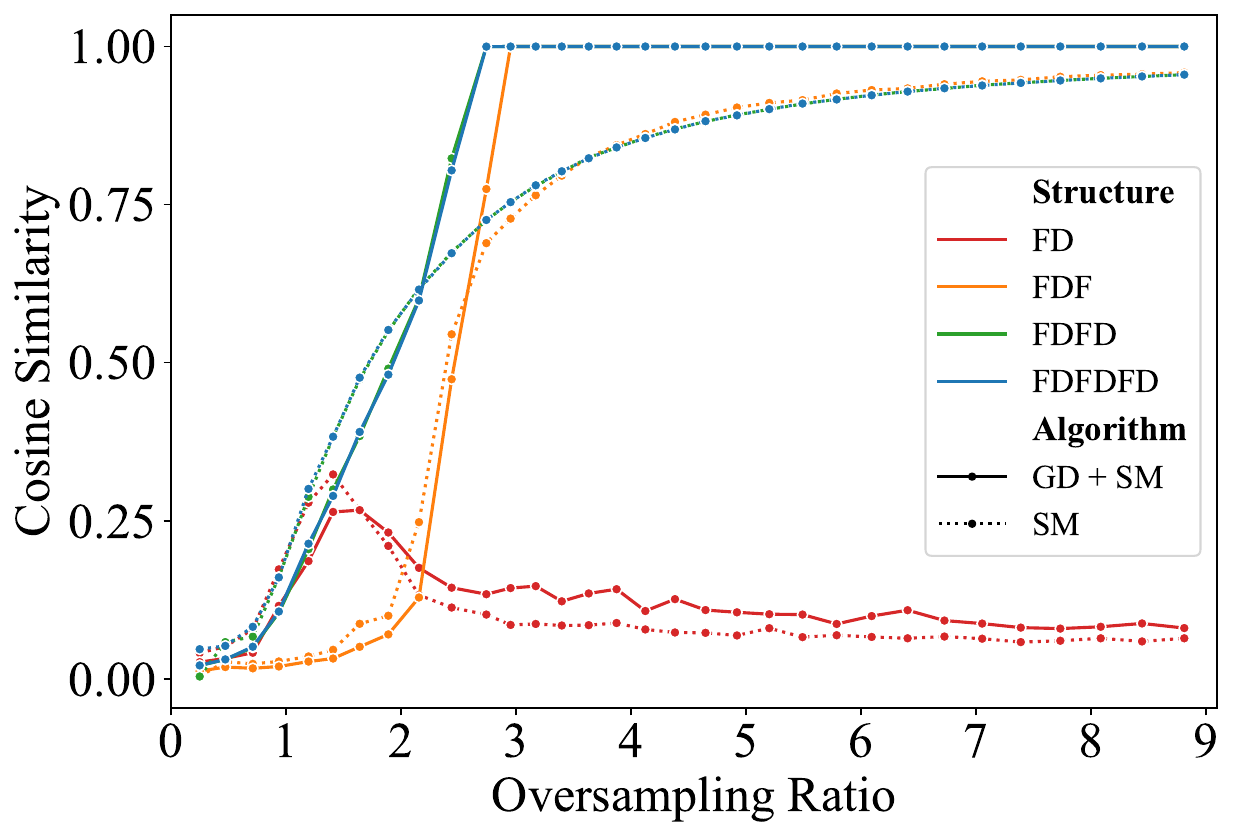}
	\caption{Performance of structured random models with different depths. A 1-layer model is unable to reconstruct, regardless of its OR, while a 1.5-layer model yields suboptimal performance. A model with more than 2 layers produces the same performance as a dense random model.
	}
	\label{fig:layers}
\end{figure}
As the complexity of the imaging system grows with the model depth, we investigate how the reconstruction performance is affected by the number of structured layers. Figure~\ref{fig:layers} illustrates the performance of the model with different depths.
For all models, the amplitude of the rightmost diagonal elements is sampled according to the Marchenko-Pastur distribution, while all the other diagonal elements have a unit amplitude.
We conclude from the results that a 1-layer model fails to produce informative measurements for meaningful recovery, regardless of the OR. Meanwhile, the 1.5-layer model enables reasonable reconstructions, although its performance consistently falls short compared to the 2-layer case across all ORs. Furthermore, all models with 2 or more layers yield a performance equivalent to the i.i.d. Gaussian case, as seen in Figure~\ref{fig:algo}. These observations suggest that an optical system with 2 layers is sufficient to emulate i.i.d. random models in practice.

\subsection{Input-Dependent Reconstruction of One and a Half Layers}\label{sec:1p5-exp}
\begin{figure}[t!]
	\centering
	\includegraphics[width=0.95\linewidth]{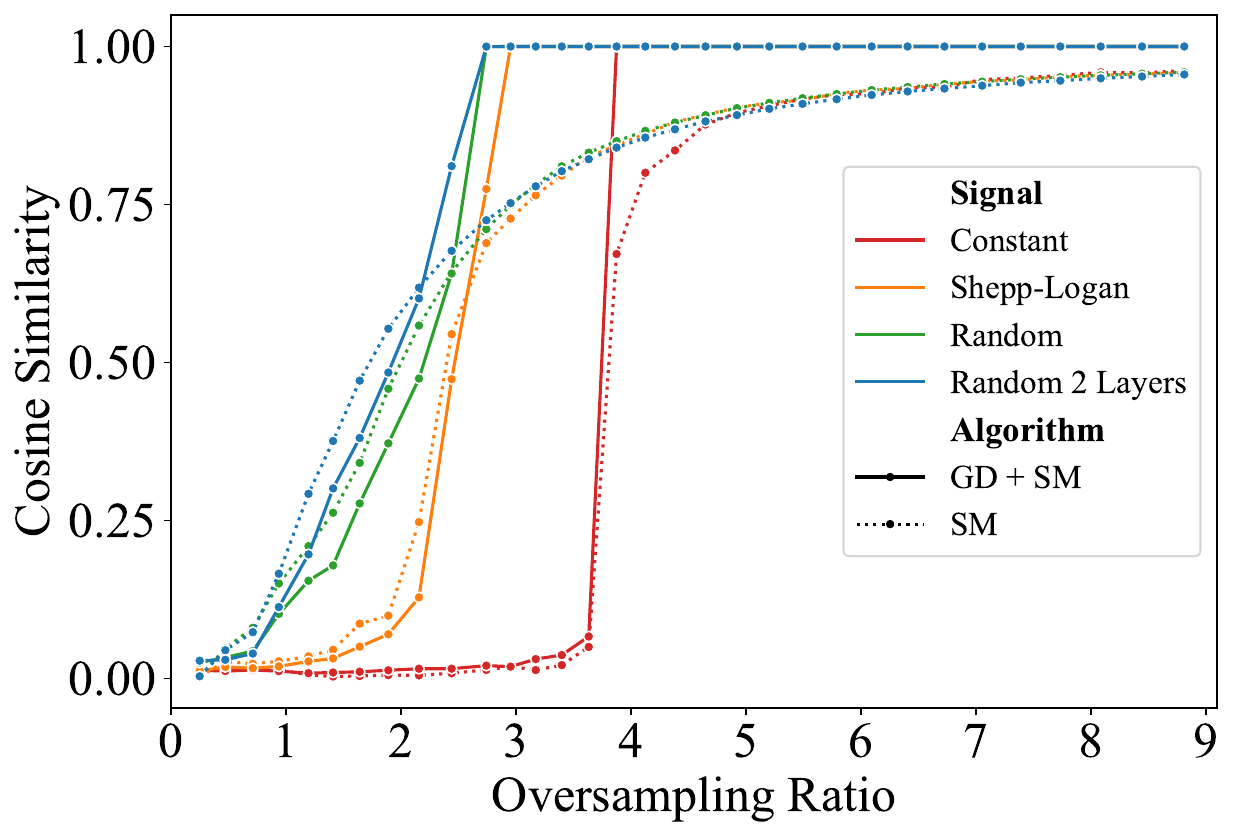}
	\caption{Performance comparison between different signals with the 1.5-layer structured random model. The performance improves when the input signal has a larger phase variance. The performance approaches that of a 2-layer model when the input signal is random.}
	\label{fig:1-5layer}
\end{figure}
As shown in Section~\ref{sec:cov-1p5}, the covariance pattern of the sensing matrix resembles a left-circulant matrix.
The reconstruction in this case depends on the input signal to recover.
We report in Figure~\ref{fig:1-5layer} the reconstruction performance of several signals in the 1.5 layer case.
Besides the Shepp-Logan signal, we use a constant signal where all phases are zero, and a random signal whose phases are uniformly distributed in $[-\pi, \pi)$. We conclude from Figure~\ref{fig:1-5layer} that a better performance can be obtained if one increases the variance of the phase distribution of the signal. Specifically, the performance for a random signal with 1.5 layers is close to that with 2 layers.
We explain this behavior by expressing the signal in the Fourier domain as $\signal = \trans^\ctrans \tilde\signal \trans$, such that the measurements become $\measure = |\trans \diag_1 \tilde\signal \trans|^2$.
Ideally, one would like $\diag_1 \tilde\signal$ to retain all the random elements of $\diag_1$.
However, for a constant signal---or, more generally, for weak-phase objects---the Fourier transform $\tilde\signal$ is sparse and the measurements $\measure$ depend on very few diagonal elements of $\diag_1$.

\subsection{Binary Phase Distribution} \label{sec:alt-phase}
To explore the generality of our structured random model, we experiment with three phase distributions for the phase $\theta_j$ in the diagonal elements $d_j = r_j \mathrm{e}^{\mathrm{j} \theta_j}$:
\begin{enumerate}
	\item Uniform: uniform distribution on $[-\pi, \pi)$;
	\item Laplace: distribution of $\operatorname{arctan2}(y,x)$, where $x$ and $y$ are drawn from a standard Laplace distribution;
	\item Binary: uniform distribution in $\{0, \pi \}$.
\end{enumerate}

\begin{figure}[t!]
	\centering
	\begin{subfigure}{0.95\linewidth}
		\centering
		\includegraphics[width=1.0\textwidth]{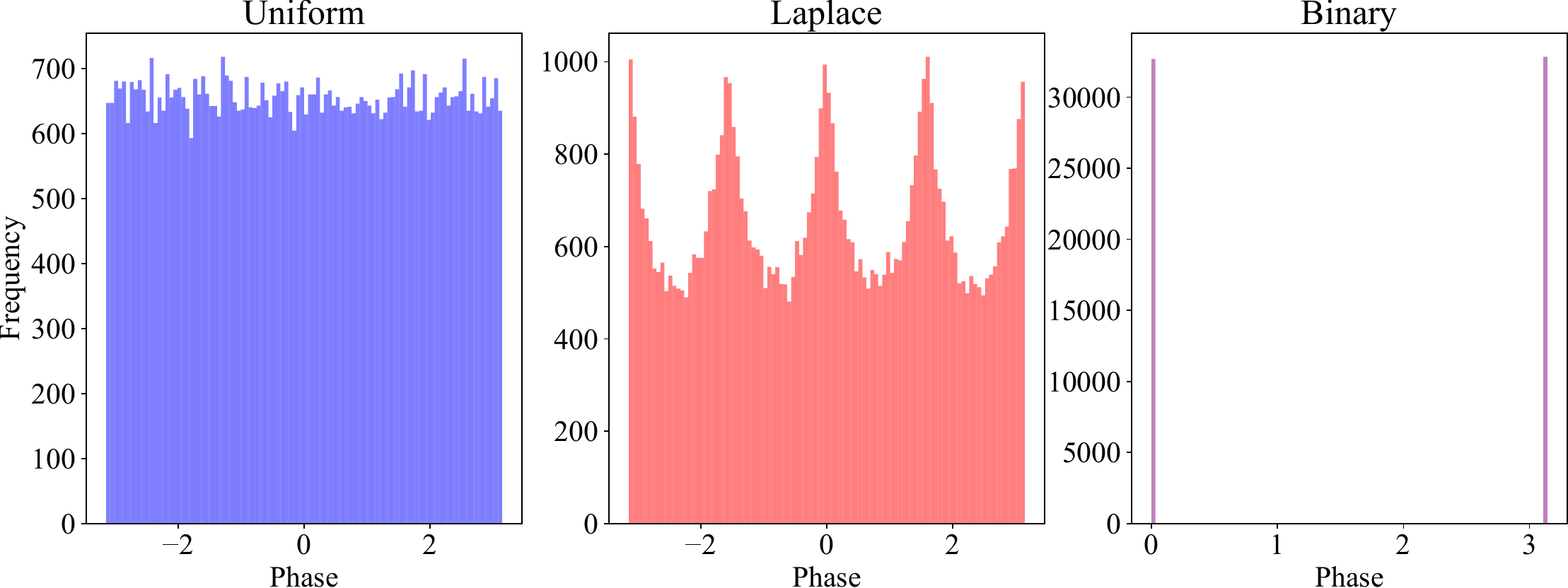}
		\caption{Alternative phase distributions}
		\label{fig:diag_distribution}
		\vspace{15pt}
	\end{subfigure}
	\centering
	\begin{subfigure}{0.95\linewidth}
		\centering
		\includegraphics[width=1.0\textwidth]{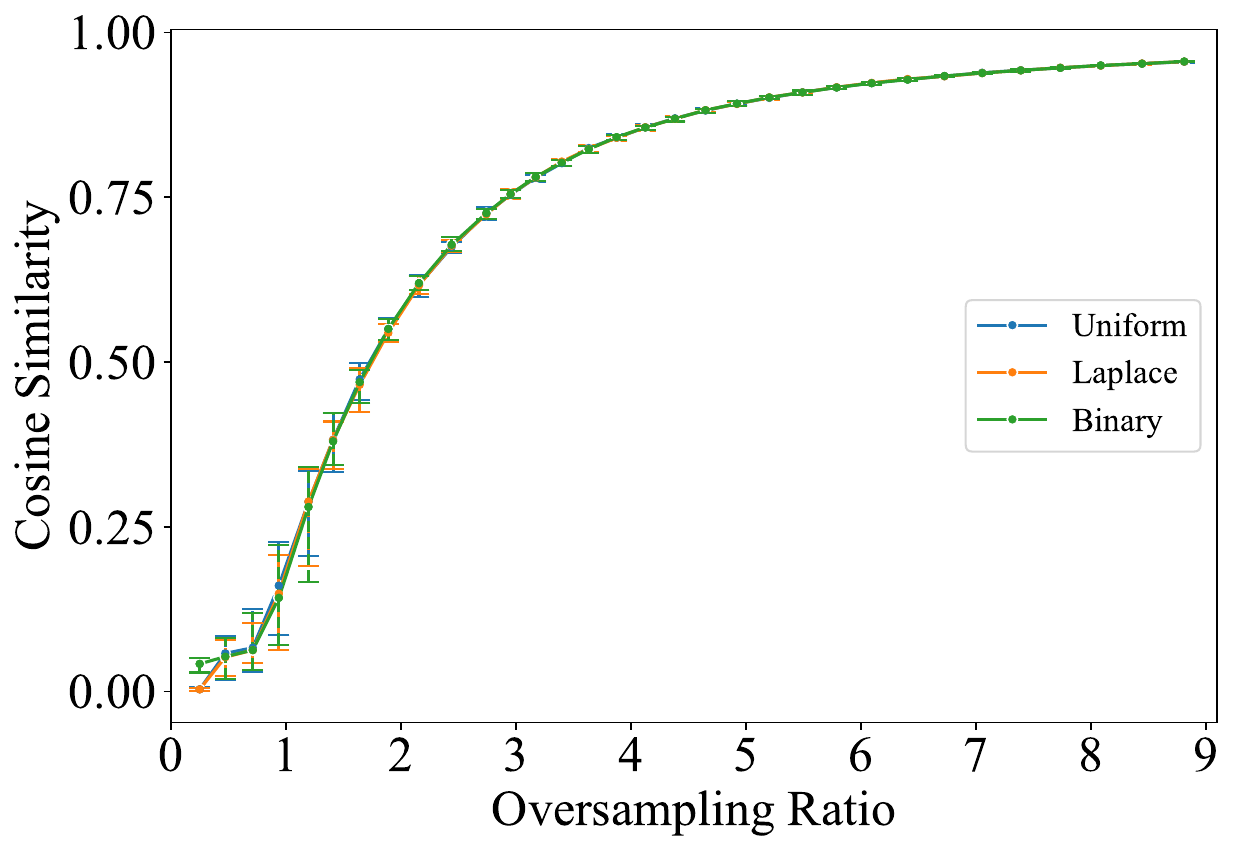}
		\caption{Performance of the alternative phase distributions}
		\label{fig:phase}
	\end{subfigure}
	\caption{Performance comparison for different phase distributions, including the uniform, Laplace, and binary distributions. Fig. (a) depicts the distributional difference between the alternatives. Fig. (b) shows that the reconstruction performance remains identical despite the distributional difference.}
	\label{fig:uni-phase}
\end{figure}
We provide in Figure~\ref{fig:diag_distribution} a histogram of a realization of a phase that follows these three distributions and observe that they are substantially different.
Despite this, we see in Figure~\ref{fig:phase} that all options produce the same result. The binary case is of
particular practical interest because such diffusers can be manufactured
by single-step photolithography on a transparent substrate, etched at
only two depth levels. This is significantly simpler than fabricating and calibrating continuous phase profiles. Additionally, the diagonal matrices at different layers can share identical elements, which would allow all diffusers to be manufactured from the same template.

The observed universality can be understood if one considers that the diagonal elements are effectively summed through the Fourier transform. Due to the central limit theorem, this summation leads the resulting matrix elements to approach a complex-valued Gaussian distribution, independently of the specific distribution of the diagonal entries.
\subsection{Robustness to Noise}
To examine the robustness of the structured random model, we apply additive white Gaussian noise to the clean measurement and set
\begin{equation}
	\measure = \max(| \sensing \signal |^2 + \boldsymbol{\epsilon}, 0),
\end{equation}
where $\boldsymbol{\epsilon} \sim \mathcal{N}(\mathbf{0}, \sigma_{\boldsymbol{\epsilon}}^2 \mathbf{I})$ with $\sigma_{\boldsymbol{\epsilon}} = \eta \mu_{| \sensing \signal |^2}$ and $\eta \in [0, 1]$ controls the strength of the noise relative to the average of the clean measurements $\mu_{| \sensing \signal |^2}$.
\begin{figure}
	\centering
	\includegraphics[width=0.95\linewidth]{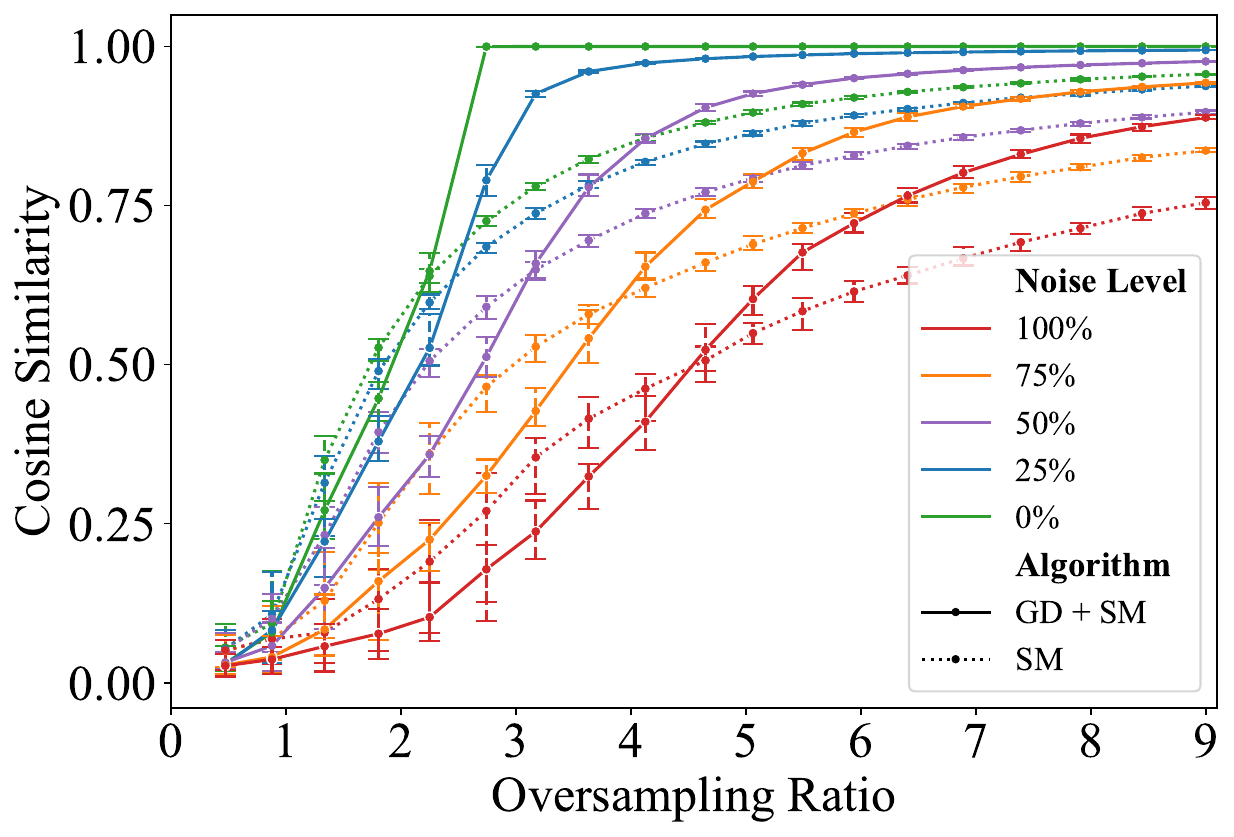}
	\caption{Performance of the 2-layer structured random model under different noise levels. A good reconstruction is maintained at an OR of 4 under a noise level of 50\% of the average of the clean measurements.
	}
	\label{fig:noise}
\end{figure}
Figure~\ref{fig:noise} displays the reconstruction performance of a 2-layer structured random model under different noise levels. The model exhibits strong robustness to noise, maintaining a steady reconstruction even with a noise of 50\% of the average of the clean measurements. At higher noise levels, the performance at typical ORs (3\,--\,4) is considerably reduced but can be improved by increasing the OR to 6\,--\,7.

% To summarize, we have shown that a 2-layer structured random model enables the same reconstruction performance as the i.i.d. random model, while reducing computational complexity from quadratic to log-linear. The diagonal matrices require only a binary phase distribution and may reuse identical elements, and the model demonstrates strong robustness to measurement noise. These features together enable an efficient optical implementation of the structured random model.
% In this section, we explore the optimal architecture for the structured random model and conclude that 2 layers are necessary and sufficient to obtain the best performance.
% To explain this behavior, we study the covariance and show that performance is linked to the correlation between the final matrix elements.
% Additionally, we provide evidence that the model's performance is determined by its spectrum, providing a convenient way to emulate any class of right unitary-invariant random matrices.
% The structured random model also exhibits strong robustness in reconstruction under the presence of measurement noise.

\subsection{Spectrum}\label{sec:spectrum}
\begin{figure}[t!]
	\centering
	\includegraphics[width=0.95\linewidth]{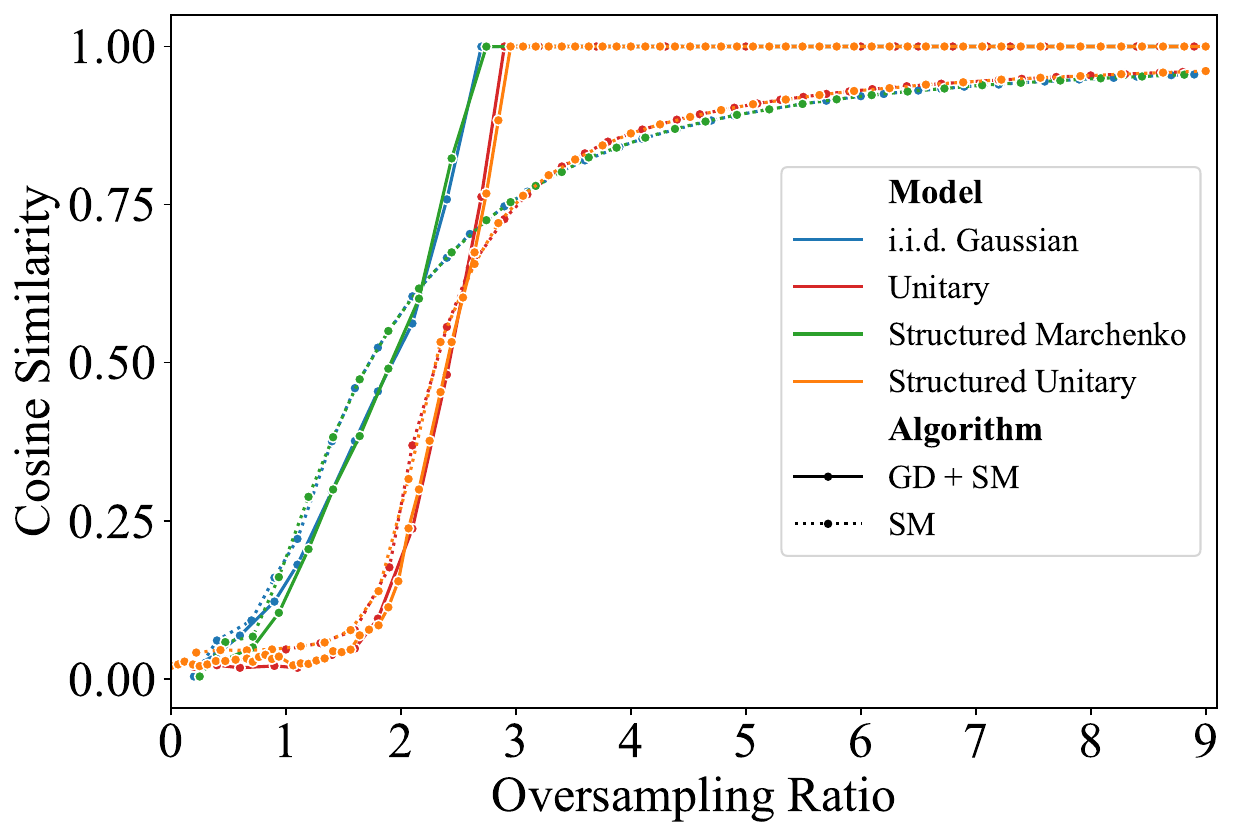}
	\caption{Performance comparison between the structured and dense random models under different singular spectra. Structured models can yield the same performance as dense models with the same spectrum.}
	\label{fig:spectrum}
\end{figure}
We observe that the spectrum is the determining factor that characterizes the performance of structured and dense random models. For the 2-layer forward matrix $\sensing = \trans_1 \diag_1 \trans_2 \diag_2 \overs$ with OR\,$>1$ and $\diag_1$ a unit-magnitude configuration, the spectrum can be extracted through the Gramian
\begin{align*}
	\sensing^\ctrans \sensing & =  \overs^\ctrans \diag_2^\ctrans \trans_2^\ctrans \diag_1^\ctrans \trans_1^\ctrans \trans_1 \diag_1 \trans_2 \diag_2 \overs \\
	                          & = \operatorname{diag}(\mathbf{r}_2^2)_{1:n}, \quad \text{($\trans_1, \diag_1, \trans_2$ are all unitary)}
\end{align*}
where $\mathbf{r}_2$ represents the magnitude of the diagonal elements of $\diag_2$ and $\operatorname{diag}(\cdot)_{1:n}$ denotes the ($n \times n$) part of a diagonal matrix on the top-left corner. This shows that the eigenvalues of the Gramian are the first $n$ elements of $\mathbf{r}_2^2$, which are equal to the squared spectrum of the 2-layer model. Therefore, a 2-layer model can conveniently emulate dense Gaussian or unitary matrices, with the Marchenko-Pastur distribution or units for $\mathbf{a}_2$, respectively. We compare in Figure \ref{fig:spectrum} the performance of the structured random models with Marchenko-Pastur distribution for $\diag_2$ or all unitary diagonals with their dense counterparts.
We observe that the structured setups match the performance of the dense models. This also explains the low dependence of the performance on the phase distribution, since the latter does not alter the spectrum.

In the undersampled case, the spectrum is obtained from $\sensing \sensing^\ctrans$, which retains the same nonzero support as the sampling distribution of $\mathbf{r}_2$ beside the expected ($n-m$) vanishing singular values, as guaranteed by the Cauchy interlacing theorem~\cite{magnus2019matrix}. This distributional shift has little effect on overall performance.

The direct configuration also generalizes to the spectrum of any matrix, for instance, products of i.i.d. Gaussian matrices. A product of $N$ random Gaussian matrices, $\sensing \in \complex^{m \times n}$, is defined as
\begin{equation}
	\sensing = \prod_{i=1}^N \mathbf{X}_i,
\end{equation}
where $\mathbf{X}_N \in \complex^{m \times n}$ and $\mathbf{X}_i \in \complex^{m \times m}, \forall i \in [1,\ldots,N-~1]$ are i.i.d. sampled from a complex Gaussian distribution $\mathcal{CN}(0,1/m)$.
\begin{figure}[t!]
	\centering
	\begin{subfigure}{0.9\linewidth}
		\centering
		\includegraphics[width=1.0\linewidth]{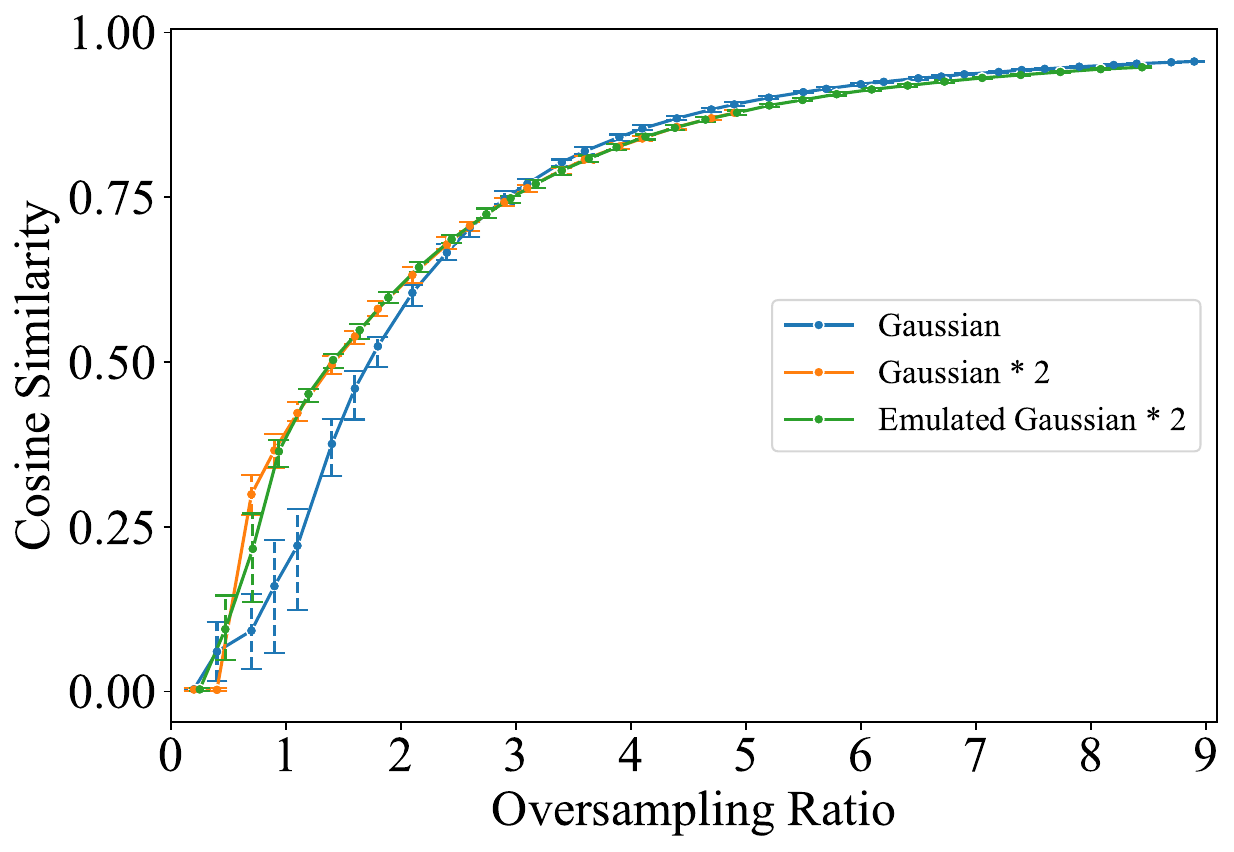}
		\caption{2 Gaussians}
		\label{fig:control2}
		\vspace{15pt}
	\end{subfigure}

	\begin{subfigure}{0.9\linewidth}
		\centering
		\includegraphics[width=1.0\linewidth]{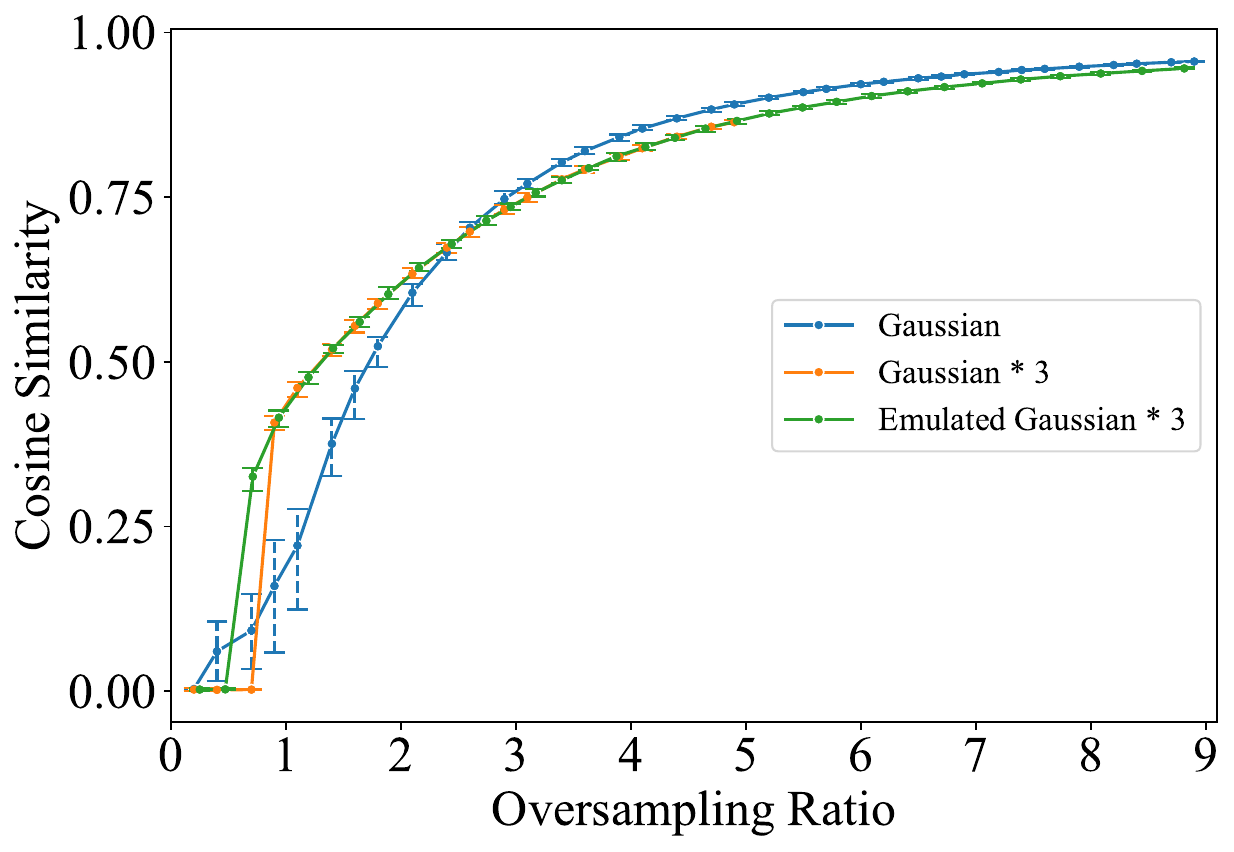}
		\caption{3 Gaussians}
		\label{fig:control3}
		\vspace{15pt}
	\end{subfigure}

	\begin{subfigure}{0.9\linewidth}
		\centering
		\includegraphics[width=1.0\linewidth]{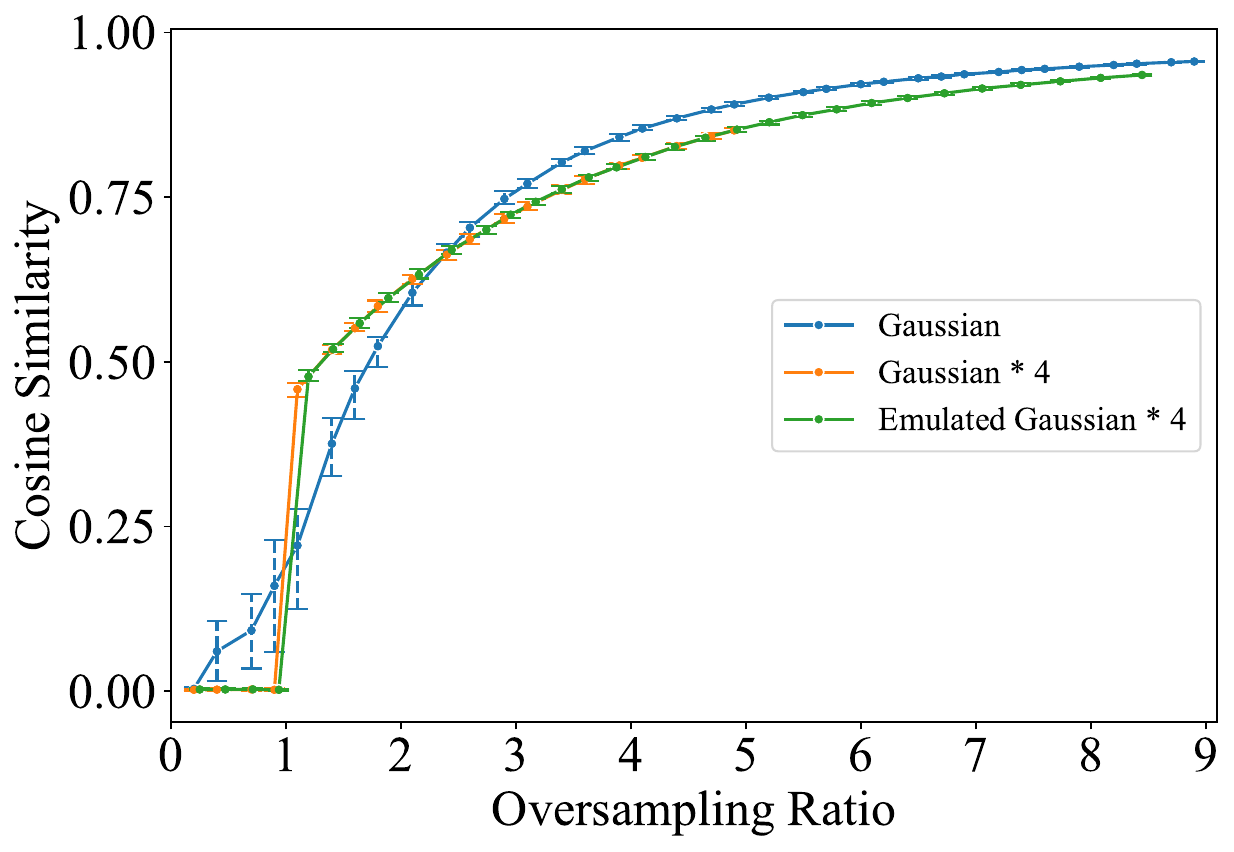}
		\caption{4 Gaussians}
		\label{fig:control4}
	\end{subfigure}

	\caption{Performance of products of Gaussian models. Structured models can emulate dense models by using the same spectrum.}
	\label{fig:control}
\end{figure}
We compare in Figure \ref{fig:control} the performance of the dense random model using a product of 2 to 4 Gaussian matrices to that of their structured random counterparts using the spectra from the dense models. We conclude that the structured random model can lead to the same results as the dense random model when oversampling. Notably, the results are better than those of i.i.d. models at low ORs. When OR\,$<$\,1, the performance differs as undersampling alters the spectrum.
This empirical observation parallels the spectral universality rigorously
shown for linear inverse problems in \cite{wang2024universality,
dudeja2024spectral}, which our nonlinear model reaches only with two or more structured layers.
% It further generalizes the use of a structured random model to any dense random matrix, as they need no more than the sole knowledge of the spectrum.
\subsection{Algorithm Comparison} \label{sec:algorithm-comparison}
We compare several common algorithm choices to further refine the initialization provided by spectral methods.
These include steepest gradient descent:
\begin{equation}
	\hat{\signal}_{t+1} = \hat{\signal}_{t} - \gamma \nabla_{\hat{\signal}} \mathcal{L} (\hat{\signal}_{t}),
\end{equation}
using either the $\ell_2$ loss $\mathcal{L}(\hat{\signal}) = \| |\sensing\hat{\signal}|^2 - \measure \|^2$ or the amplitude loss $\mathcal{L}(\hat{\signal}) = \| |\sensing\hat{\signal}| - \sqrt{\measure} \|^2$.

We also evaluate the measurement projection step characteristic of the Gerchberg-Saxton algorithm~\cite{gerchberg1972practical}:
\begin{equation}
	\hat{\signal}_{t+1} = \sensing'(\sqrt{\measure} \odot e^{\j\angle\sensing\hat{\signal}_{t}}),
\end{equation}
where $\angle$ denotes the element-wise phase extraction of a complex vector, and $\odot$ represents the Hadamard product. The backward projection operator $\sensing'$ is typically defined as the Hermitian conjugate:
\begin{equation}
	\sensing^\ctrans = \overs^\ctrans (\trans^s)^\ctrans \left( \prod_{i=N}^{1} \diag_i^\ctrans \trans^\ctrans \right) \unders^\ctrans.
\end{equation}
However, because our structured random model is not necessarily unitary, we propose an alternative formulation for the backward projection. We define a generalized inverse $\sensing^{-}$ where each $\diag_i^\ctrans$ in the Hermitian sequence is replaced with its inverse $\diag_i^{-1}$:
\begin{equation} \label{eq:measurement-projection}
	\sensing^{-} = \overs^\ctrans (\trans^s)^\ctrans \left( \prod_{i=N}^{1} \diag_i^{-1} \trans^\ctrans \right) \unders^\ctrans.
\end{equation}
This formulation rigorously satisfies the generalized inverse condition, $\sensing \sensing^{-} \sensing = \sensing$, provided that the system relies exclusively on either undersampling ($\overs = \mathbf{I}$) or oversampling ($\unders = \mathbf{I}$). Physically, iteration~\ref{eq:measurement-projection} can be interpreted as propagating the estimated signal into the measurement domain, enforcing the observed magnitudes, and projecting the corrected signal back into the signal domain.

Lastly, we apply a Nesterov accelerated gradient (NAG) algorithm~\cite{nesterov1983method} to the $\ell_2$ loss, utilizing the following iteration:
\begin{align}
	{\hat{\signal}'}_{t+1} & = \hat{\signal}_{t} - \gamma \nabla_{\hat{\signal}} \mathcal{L} (\hat{\signal}_{t}),                        \\
	\hat{\signal}_{t+1}    & = {\hat{\signal}'}_{t+1} + \frac{\beta_{t}-1}{\beta_{t+1}} ({\hat{\signal}'}_{t+1} - {\hat{\signal}'}_{t}),
\end{align}
with the momentum parameter updated as $\beta_{t+1} = \frac{1 + \sqrt{1+4\beta_{t}^2}}{2}$ and initialized at $\beta_1 = 1$.
\begin{figure}[t!]
	\centering
	\includegraphics[width=0.95\linewidth]{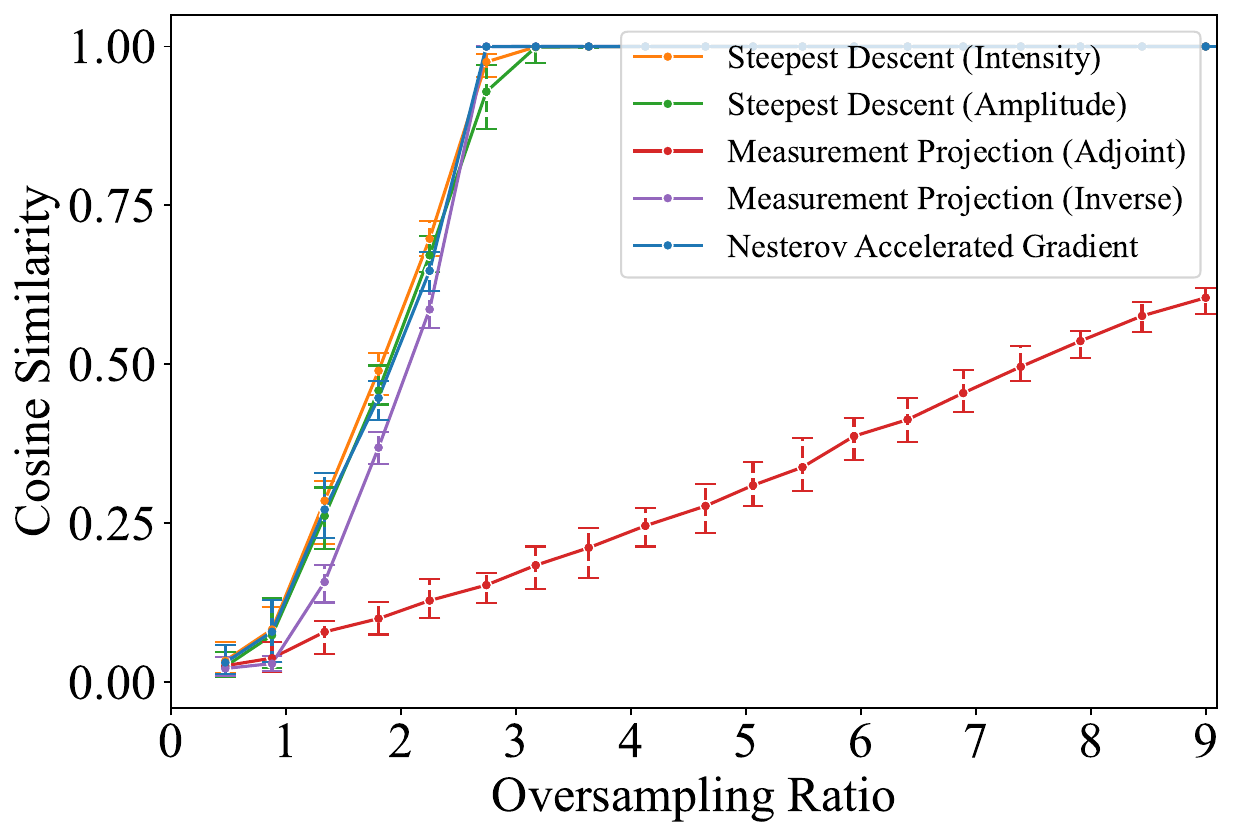}
	\caption{Performance comparison between different algorithms after spectral initialization. Nesterov accelerated gradient achieves full recovery threshold with the lowest oversampling ratio of 2.5.}
	\label{fig:comparison}
\end{figure}

Figure~\ref{fig:comparison} illustrates the performance comparison of these algorithms. We evaluate with a 2-layer Marchenko-Pastur configuration defined in section~\ref{sec:model-setup} without measurement noise. First, we observe that steepest descent with amplitude loss performs slightly worse than with the $\ell_2$ loss, despite amplitude loss generally being considered more robust in practical settings~\cite{yeh2015experimental}. Moreover, standard measurement projection utilizing the Hermitian conjugate $\sensing^\ctrans$ yields poor reconstruction quality. In contrast, projecting via our proposed generalized inverse achieves performance comparable to the gradient-based methods, confirming that the discrepancy arises directly from the non-unitarity of the model. Ultimately, the NAG algorithm achieves exact recovery at the lowest oversampling ratio of 2.5, despite being slightly worse than steepest descent on lower oversampling ratios. In general, all effective algorithms yield similar performance curves.

\subsection{Alternative Transforms}
In addition to the FFT, in this section, we investigate two more unitary structured transforms: the discrete cosine transform (DCT) and the Hadamard transform (HT). Like the FFT, these transforms can also be evaluated in $\mathcal{O}(n \log n)$ time, and can be treated as real-valued proxies for the FFT.
\begin{figure}[t!]
	\begin{subfigure}{0.9\linewidth}
		\centering
		\includegraphics[width=1.0\linewidth]{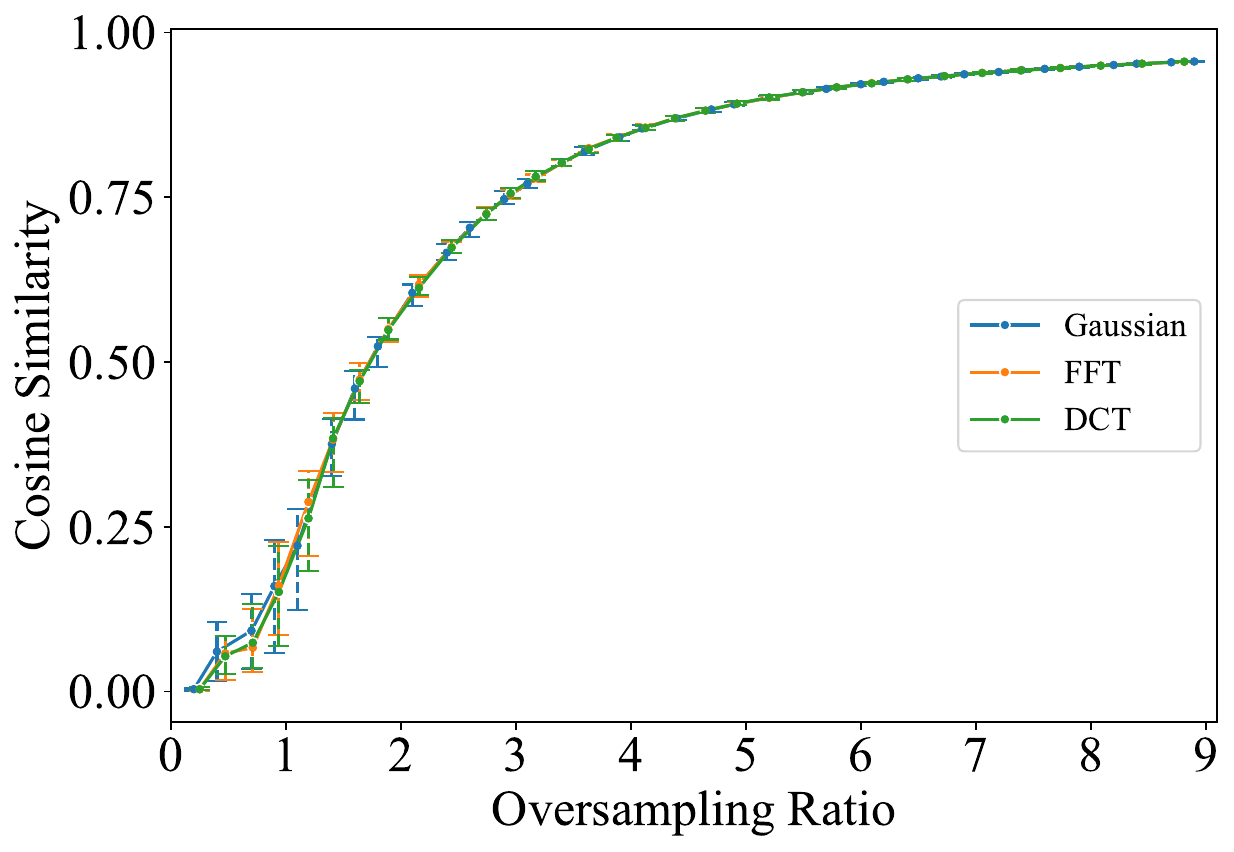}
		\caption{Alternative transforms;}
		\label{fig:trans}
		\vspace{15pt}
	\end{subfigure}

	\begin{subfigure}{0.9\linewidth}
		\centering
		\includegraphics[width=1.0\linewidth]{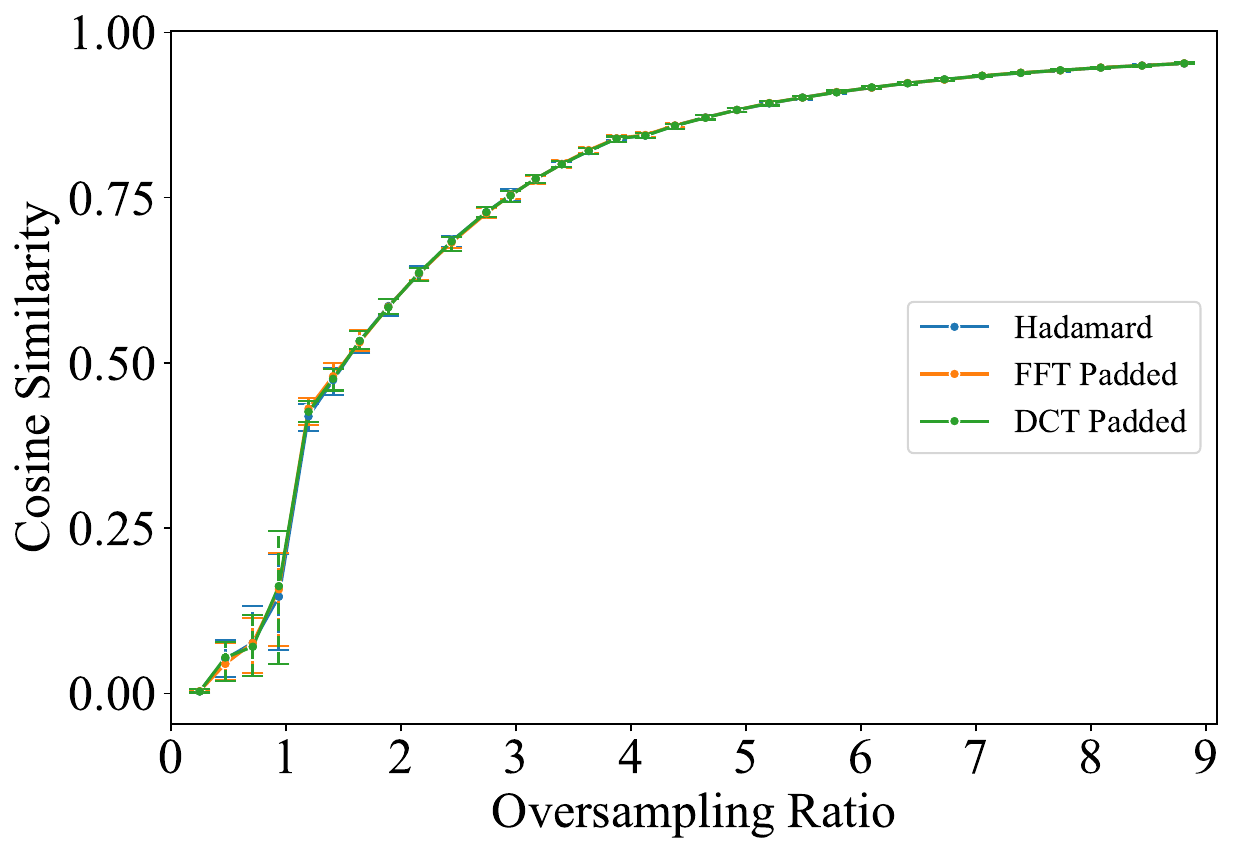}
		\caption{Alternative transforms with additional padding;}
		\label{fig:trans_pad}
	\end{subfigure}

	\caption{Performance comparison of the models using FFT, DCT and HT. All transforms yield the same performance.}
	\label{fig:uni-trans}
\end{figure}
We show in Figure~\ref{fig:trans} the results of the FFT and DCT. Both transforms yield identical results. For the HT, additional padding is needed as the HT only admits the input dimension to be $p = 2^k, k \in \mathbb{N}$. To control this factor, we apply the same over- and undersampling for the FFT, DCT, and HT, which leads to the performance that we report in Figure~\ref{fig:trans_pad}.
All employed transforms generate identical results.

\section{Conclusion} \label{sec:outro}
We have proposed to replace dense random models, traditionally used in phase retrieval, with a structured random model.
This allows us to obtain the same reconstruction accuracy as dense models at a fraction of their original computational costs. Our model is straightforward to realize with practical lenses and diffusers. The diffusers need only a random binary phase pattern and can share the same realization. Strong robustness to the measurement noise has also been demonstrated. For an optimal architecture, we have shown that two layers are the necessary and sufficient setup for the optimal performance. They are required to generate uncorrelated and asymptotically Gaussian-distributed elements for the overall forward matrix. Further, we have shown that the spectrum of the forward matrix determines the performance of the model. By configuring the spectrum through certain diagonal matrices, we can emulate any dense model.
As a result, alternative phase distributions and structured transforms can be used to obtain the same performance, as they maintain the same spectrum.
% For future directions, a critical next step is the physical implementation and empirical validation of the proposed optical system.
% This would necessarily include rigorous calibration of the diffusers, for which existing calibration methods~\cite{popoff2010measuring, boniface2020non, monga2021algorithm} could be adapted and evaluated.
% From a theoretical perspective, several open questions remain. First, determining the optimal spectrum configuration as a function of oversampling ratio would provide valuable insights for system design. Second, establishing a rigorous mathematical proof of the equivalence between dense and structured random models would strengthen the theoretical foundations of this approach, beyond the covariance analysis presented here.
Our work contributes to the growing body of applications that leverage efficient transforms~\cite{ailon2009fast, dong2020reservoir, yu2016orthogonal}, and holds a significant potential for computational accelerations across diverse domains. Finally, our proposed model provides a natural testbed for the exploration of more sophisticated reconstruction algorithms~\cite{ducotterd2025undersampled}, including accelerated first-order methods with classic regularization~\cite{beck2009fast,chambolle2016introduction} and learned regularization approaches~\cite{jin2017deep,devalla2018drunet}.

\bibliography{refs}

\bibliographystyle{IEEEtran}

\end{document}